\documentclass[a4paper,11pt]{article}

\usepackage[dvipsnames]{xcolor}
\usepackage{amsfonts,amsthm,amssymb}
\usepackage{dsfont}
\usepackage{tikz}
\usepackage{setspace}
\usepackage{subfigure}
\usepackage[flushleft]{threeparttable}
\usepackage{endnotes}

\usepackage{bm}
\usepackage{mathtools}
\mathtoolsset{showonlyrefs=true}

\usepackage[authoryear]{natbib}

\usepackage[plainpages=false, pdfpagelabels]{hyperref} 
\hypersetup{
	colorlinks   = true,
	citecolor    = RoyalBlue, 
	linkcolor    = RubineRed, 
	urlcolor     = Turquoise
}

\newcommand\COMP{\hbox{C\kern -.58em {\raise .54ex \hbox{$\scriptscriptstyle |$}}
\kern-.55em {\raise .53ex \hbox{$\scriptscriptstyle |$}} }}
\newcommand\NN{\hbox{I\kern-.2em\hbox{N}}}
\newcommand\RR{\hbox{I\kern-.2em\hbox{R}}}
\newcommand\sRR{{\it \hbox{I\kern-.2em\hbox{R}}}}
\newcommand\QQ{\hbox{I\kern-.53em\hbox{Q}}}
\newcommand\PP{\hbox{I\kern-.53em\hbox{P}}}
\newcommand\EE{\hbox{I\kern-.53em\hbox{E}}}
\newcommand\ZZ{{{\rm Z}\kern-.28em{\rm Z}}}
\newcommand\be{\begin{equation}}
\newcommand\ee{\end{equation}}

\newtheorem{theorem}{Theorem}[section]
\newtheorem{assum}[theorem]{Assumption}
\newtheorem{proposition}[theorem]{Proposition}
\newtheorem{remark}[theorem]{Remark}

\newtheorem{lemma}[theorem]{Lemma}

\newtheorem{definition}[theorem]{Definition}
\newtheorem{corollary}[theorem]{Corollary}

\newcommand\beq{\begin{eqnarray}}
\newcommand\eeq{\end{eqnarray}}
\newcommand\bq{\begin{eqnarray*}}
\newcommand\eq{\end{eqnarray*}}

\numberwithin{equation}{section}

\newcommand{\N}{\mathcal{N}}
\newcommand{\dsone}{\mathds{1}}
\newcommand{\F}{\mathcal{F}}
\newcommand{\Fb}{\mathbb{F}}
\newcommand{\Po}{\mathbb{P}}
\newcommand{\dd}{\mathrm{d}}
\newcommand{\Nb}{\mathbb{N}}
\newcommand{\Lc}{\mathcal{L}}

\newcommand{\lamb}{\bm{\lambda}}
\newcommand{\Lamb}{\bm{\Lambda}}
\newcommand{\Mb}{\bm{M}}
\newcommand{\Mhat}{\widehat{M}}

\allowdisplaybreaks

\newcommand{\blue}{\color{blue}}

\begin{document}
\title{An Explicit Default Contagion Model and Its Application to Credit Derivatives Pricing
\footnote{We would like to thank Tahir Choulli, Liuren Wu, Zhengyu Cui, Ping Li, Shiqi Song, Xianhua Peng, Chao Shi for their insightful comments. The research of Jun Deng is supported by the National Natural Science Foundation of China (11501105) and UIBE Excellent Young Research Funding (302/871703).
The research of Bin Zou is supported by a start-up grant from the University of Connecticut. 
Declarations of interest: none}
}

\author{
Dianfa Chen\thanks{School of Finance, Nankai University, Tianjin, China. Email: dfchen@nankai.edu.cn.}, \quad
Jun Deng\thanks{School of Banking and Finance, University of International Business and Economics, Beijing, China. Email: jundeng@uibe.edu.cn} , \quad
Jianfen Feng\thanks{School of Banking and Finance, University of International Business and Economics, Beijing, China.  Email: \rm{danxin97@163.com}} , \quad
Bin Zou\thanks{Corresponding author. 341 Mansfield Road U1009
 Storrs, CT 06269-1069, Department of Mathematics, University of Connecticut, Storrs, CT, USA. Email: bin.zou@uconn.edu}
}


\maketitle

\begin{abstract}
	We propose a novel credit default model that takes into account the impact of macroeconomic information and contagion effect on the defaults of obligors.
	We use a set-valued Markov chain to model the default process, which is the set of all defaulted obligors in the group.
	We obtain analytic characterizations for the default process, and use them to derive pricing formulas in explicit forms for synthetic collateralized debt obligations (CDOs).
	Furthermore, we use market data to calibrate the model and conduct numerical studies on the tranche spreads of CDOs.
	We find evidence to support that systematic default risk coupled with  default contagion could have the leading component of the total default risk.
	
\end{abstract}

\section{Introduction}

In the aftermath of the financial crisis of 2007-2009, there have been burgeoning interests in studying the causes and remedies of this crisis from academic scholars, practitioners, and regulators around the world.
\cite{fcic2011} concludes that this financial crisis is avoidable and mainly caused by the failure of financial regulations and supervisions, especially on structured finance products.
The  collapsing mortgage-lending standards and mortgage related financial products, such as mortgage-backed securities (MBS) and collateralized debt obligations (CDOs),
lit and spread the flame of contagion and crisis.
In Basel Accords II, the calibration  of default risk neglects the significant impact of default contagion among different obligors.
After the financial crisis, Basel Accords III is proposed, and one of its key principles is to emphasize the modeling of default contagion, which contributes significantly to the collapse of financial systems during the crisis.

In the literature, there are three main strands on the modeling of default risk. 
We review them briefly in what follows, but are by no means exhaustive here.  
The first line includes the structure models, pioneered by   \cite{merton1974pricing},  which follows from the   option pricing theory of \cite{black1973pricing}. 
A  default event occurs if the company's asset value is below its debt  at the time maturity.  \cite{black1976valuing} extends Merton's structure model by  postulating  that a default occurs at the first passage time when the firm's asset value drops below a certain time-dependent barrier. 
However, structure models lack accuracy in
explaining the cross-section of credit spreads, measured by the yield difference between
risky corporate bonds and riskless bonds,  and  underpredict short-term default probabilities, see, e.g.,  \cite{eom2004structural}.   We refer to \cite{sundaresan2013review} for an excellent review on  structure models.  The second strand is  the Copula models,  first  proposed  by  \cite{li1999}.
In the seminal work of \cite{li1999}, the author uses Gaussian Copula to model the default correlations and   joint distribution, and  applies the results to credit derivative pricing and hedging. However, Copula models do not fit well with the market data and have difficulty in explaining model parameters. 
Representative works using copula for credit modeling include, but certainly are not limited to, 
\cite{frey2001}, \cite{schonbucher2001},   \cite{laurent2005}, and \cite{hullwhite2006}.

Our paper falls into the third strand, which includes the intensity based models. In the intensity based credit modeling literature, two well-established modeling approaches are the \emph{top-down} approach and the \emph{bottom-up} approach, both of them are capable of fitting the market data.\footnote{More mathematical details about these two methods are presented in Section \ref{secPreliminary}.} 
In the top-down approach, models are built for the cumulative default intensity of the whole portfolio, without specifying the underlying single obligor. 
While in comparison, under the bottom-up approach, models are constructed with specified individual  default intensity.
The default time, either for the whole portfolio or individual obligor,  is usually modeled as the first jump time of a process, such as Cox process or doubly-stochastic Poisson process. Top-down models are introduced and investigated in \cite{EKGB2007},  \cite{errais2010}, \cite{giesecke2010}, and \cite{rama2013},  among others. Under the framework of bottom-up approach, single name default intensity based models are introduced in  \cite{jarrowTurbul1995},  \cite{lando1998}, and \cite{DuffieSingleton1999}. 
\cite{Mortensen2005} postulates that the default risk of obligor $i$ is given in the form of  $\overline{\lambda}^i(t) = \text{constant} \cdot \lambda(t) + \lambda^i(t)$,   where $\lambda$ captures the systematic default risk component and  $\lambda^i$ captures the idiosyncratic default risk component which is assumed  to be independent of $\lambda$ and $\lambda^j$, for all obligors $j$. The default propagation and correlations are channeled by the common systematic component $\lambda$. However, both approaches can not account for the contagion effect on idiosyncratic default risk induced by systematic default risk $\lambda(t)$. \cite{DasDuffie2007}  and  \cite{duffieEHS2009} demonstrate the presence of  frailty correlated default and the incapability of doubly stochastic assumption to capture default contagion or frailty (unobservable explanatory variables that are correlated across firms).

In this paper,  we consider a group of  $N$ defaultable obligors (or names), labeled by $(O_i)_{i=1,2,\cdots, N}$, where the default of one obligor might impact the remaining obligors in the group. The default contagion arises from two components: defaulted obligors and macroeconomic factor.  Specifically, we denote by  { $\tau_n$} (a nonnegative  random variable)  the occurrence time of the $n^{th}$ default event, and 
$X=(X_t)_{t \ge 0}$ the default process, where $X_t$ is the set of all the obligors that have defaulted by time $t$. The dynamics of the macroeconomic factor are described by an exogenous process $Y=(Y_t)_{t\geq 0}$ that might affect the occurrence of default events in the group, precisely through $(\tau_n)_{n=1,2,\cdots,N}$ and $X=(X_t)_{t \ge 0}$.

Our new default framework, taking default dependence and contagion  into account,  gives explicitly the dynamics of the default process $X$ and can be tailored to price CDOs, CDX, iTraxx, et cetera.   
Before diving into technical details, we explain the leitmotif behind our treatment.  A synthetic CDO is a portfolio consisting of $N$ single-name CDS's on obligors with individual default times $\nu_1,\nu_2,\cdots, \nu_N$ and recovery rates $R_1,R_2,\cdots, R_N$. It is standard to assume that all obligors have the same nominal values, denoted by $A$. The accumulated loss { $L=(L_t)_{t \ge 0}$} is then given by
\begin{align}\label{introuductionLt}
L_t:=\sum^N_{i=1}A(1- R_i) \mathds{1}_{\{\nu_i\leq t\}}=  A\cdot R_X(t), \quad \mbox{where} \quad R_X(t) := \sum_{i\in X_t}(1-R_i).
\end{align}

Instead of modeling  default obligors individually  as in the bottom-up approach, we use a set-valued process (default process) $X$ to characterize the evolution of default events in the group.  Hence, the loss process $L$ is fully characterized by the set-valued process $X$. The interaction between the macroeconomic factor $Y$ and the default process $X$ is channeled through a conditional Markov process with default intensity $\Lamb$, which will be specified later (see Assumption \ref{assumption_intensity}). The  process $\Lamb$ has the same flavor of intensity based approach as in  \cite{jarrowTurbul1995}, \cite{lando1998}, and \cite{cpr2004}, to name a few.

This paper has several contributions to the intensity based credit risk modeling literature. First, we explicitly construct the set-valued default process $X$ through its intensity family $\Lamb$. 
As a result, our model integrates macroeconomic impact and intergroup default contagion effect dynamically, which both top-down and bottom-up models can not capture. Markov (set-valued and real-valued)  default models have been studied both theoretically and empirically by  \cite{jarrowLando1997} and \cite{bielecki09}; however our  model leads to a more tractable formulation in pricing and hedging credit derivatives. Second, we provide a closed-form pricing formula for CDOs without using matrix exponential as in \cite{bielecki09}. This gives  significant computational advantages  beyond its tractability, especially when the obligors  $N$ is large.   We illustrate this in a particular   homogeneous   contagion model where $N$ could be as large as 125. Finally, our set-valued Markov model can be easily extended and applied to study other credit derivatives such as first-to-default, $k$-th default, et cetera. We will leave the investigation of those derivatives in  future research.

The paper is organized as follows. In Section \ref{secPreliminary}, we introduce the default contagion model. In Section \ref{sectionDefualtprocess}, we provide analytic characterizations for the default process.
In Section \ref{sectionCDOpricing}, we derive the pricing formulas of credit derivatives.
In Section \ref{sectoionNumeric}, we conduct a sensitivity analysis on tranche spreads and use market data to calibrate the model.
We conclude in Section \ref{sectionConlusions}.
In Appendixes \ref{sectionconstructionmarkov} and \ref{sec_proofs}, we present technical proofs.

\section{The Setup}
\label{secPreliminary}

In this paper, we model the macroeconomic information (or factors) by an exogenous process $Y=(Y_t)_{t\geq 0}$, which is defined on a stochastic basis $(\Omega, \F,  \Fb^Y =(\F^Y_t)_{t \ge 0}, \Po)$. Here the filtration $\Fb^Y$ is taken to be the augmented filtration generated by the process $Y$, satisfying the usual hypotheses of right continuity and completeness, and $\F^Y_\infty \subseteq \F$.
The measure $\Po$ is the risk neutral probability measure associated with a constant risk-free rate $r$.
In the economy, we consider $N$ defaultable obligors, labeled as $\{O_i\}_{i \in \N}$, where $\N:= \{1,2,\cdots,N\}$.
For each obligor $O_i$, denote by $\nu_i$ its individual default time, where $i \in \N$.
If obligor $O_i$ defaults, we assume there is a proportional nominal loss of $1-R_i$.
As a standard market practice,  $R_i$ is often set to be 40\% for all $i$, see ISDA standard CDS converter specification.\footnote{\url{http://www.cdsmodel.com/cdsmodel/assets/cds-model/docs}}

In our studies, a credit derivative is a contingent claim with payoff depending on the loss process $L=(L_t)_{t \ge 0}$, which is defined by
\begin{align}
\label{eqn_L}
L_t :=\sum^N_{i=1} (1- R_i) \dsone_{\{\nu_i\leq t\}}, \quad t \ge 0,
\end{align}
where $\dsone_\cdot$ is an indicator function.
Without loss of generality, we have assumed a unitary face value for all obligors in the above definition of $L$.
In the literature, there are two standard approaches (models) for the pricing and hedging problems of credit derivatives: the \emph{bottom-up} and \emph{top-down} approaches.
The bottom-up approach specifies the intensity process $\lambda_i = (\lambda_{i}(t))_{t \ge 0}$\footnote{Throughout this paper, for a stochastic process, if the subscript is reserved for special meaning, then we write time variable $t$ in parenthesis, see, e.g., $\lambda_{i}(t)$; otherwise, we may write time variable $t$ in subscript or parenthesis exchangeably (e.g., as $L_t$ or $L(t)$).} of each obligor such that, for all $i \in \N,$
\begin{align}
\left( \mathds{1}_{\{\nu_i\leq t\}} - \int_{0}^{t}\lambda_{i}(s) \dd s \right)_{t\ge 0} \text{ is a martingale},
\end{align}
see, e.g., \cite{duffiepansingleton2000} and \cite{cpr2004}.
On the other hand, the top-down approach specifies the constituent intensity process $\lambda_L = (\lambda_L(t))_{t \ge 0}$ of the loss process $L$ such that
\begin{align}
\left( L_t - \int_{0}^{t} \lambda_{L}(s) \dd s \right)_{t\ge 0} \text{ is a martingale},
\end{align}
see, e.g., \cite{EKGB2007} and \cite{giesecke2010}.
In the top-down approach, the constituent intensity $\lambda_L$ can be recovered by random thinning to decompose the aggregate portfolio intensity into a sum of constituent intensities, see \cite{giesecke2010}.

Unlike in the { bottom-up} approach, we do not model the individual default time $\nu_i$ for each obligor $O_i$, but instead consider the \emph{ordered} default times $(\tau_i)_{i \in \N}$ of all obligors. Namely, $\tau_i$ is the occurrence time of the $i$-th  default among all obligors.
Hence, we have
\begin{align}
\tau_1 := \min_{i \in \N} \{ \nu_i \} \le \cdots \le \tau_i \le \cdots \le \tau_N := \max_{i \in \N} \{ \nu_i \}.
\end{align}
Denote by $X=(X_t)_{t \ge 0}$ the \emph{default process}, and define $X_t$ as the set of obligors that have defaulted by time $t$.
$X$ is then a set-valued process taking values in the subsets of $\N$. For instance, if $X_t = \{1, 5, 9\}$, then obligors $O_1$, $O_5$, and $O_9$ have defaulted by time $t$.
With the usual conventions, we set $\tau_0 = 0$ and $X_0 = \emptyset$.
The following   assumptions on the defaults modeling will be imposed throughout the paper.

\begin{assum}
	\label{assumption_default}
	We assume that
	(i) no more than one default occurs at the same time;
	and (ii) obligors will not recover after the default.
\end{assum}

\begin{remark}
	Under Assumption \ref{assumption_default}, we have the following results.
	\begin{itemize}
		\item[\rm (i)] $\Po$(no default at time $t$) = 1 - $\Po$(one default at time $t$) and $\tau_1 < \cdots < \tau_i < \cdots < \tau_N$.
		\item[\rm (ii)] $X$ is a non-decreasing process, i.e., $X_s \subseteq X_t$\footnote{In this paper, the notation ``$\subseteq$" may contain equality, i.e., it is possible that $E \subseteq F$ and $F \subseteq E$ hold at the same time. If $E$ is a true subset of $F$, we denote by $E \subset F$.} for all $0 \le s < t $.
		\item[\rm (iii)] The cardinality of the default process, denoted by $|X_t|$ at time $t$, jumps up by size 1 at default time $\tau_i$.
		Hence, $|X_{\tau_{i+1}} / X_{t}| = 1$, where  $\tau_i \le t < \tau_{i+1}$ and $i \in \N / \{N \}$.
		Here, $F/E$ denotes the set difference of two sets $E$ and $F$, i.e., the set of all elements that belong to $F$ but not $E$.
	\end{itemize}
\end{remark}

In our setup, the defaults of all obligors are fully characterized by the pair $(\tau_i, X_{\tau_i})_{i \in \{0\} \cup \N}$. By using them, we rewrite the default process $X$ and the loss process $L$ by
\begin{align}
X_t = \sum_{i=0}^{N} X_{\tau_i} \cdot \dsone_{\{ \tau_i \le t < \tau_{i+1} \}} \quad
\text{and} \quad
L_t =\sum_{i \in X_t} (1-R_i) :=  R_X(t), \quad t \ge 0.  \label{eqn_L_equivalent}
\end{align}
In the above expression, we take $\tau_{N+1} = +\infty$.
Since $\tau_N$ records the last default among all obligors, it is clear that $X_t = \N$ for all $t \ge \tau_N$.

One of the novelties of our framework is to characterize defaults directly through the dynamics of the default process $X$, which is modeled by an $\Fb^Y$-conditional Markov chain with intensity family $\Lamb=(\Lambda_{EF}(t))_{t \ge 0}$, where $E,F \in \Nb$.
Here, $\Nb$ denotes the sigma-algebra consisting of all the subsets of $\N$.
To account for the dependence among defaults in the group, the intensity family $\Lamb$ may depend on the macroeconomic factor $Y$ and/or intergroup contagion.
Introduce notations $\Fb^X = (\F^X_t)_{t \ge 0}$ as the augmented filtration generated by the process $X$.
We present two important definitions below.

\begin{definition}
	\label{definiconditionalmarkov}
	A continuous time ${\mathbb{N}}$-valued stochastic process $X=(X_t)_{t\geq 0}$  is called an $\Fb^Y$-conditional Markov chain if, for all $0 \le s \leq t$ and $ F \in \Nb$,  the following condition holds:
	\begin{align}
	\mathbb{P} \left( X_t= F \mid \F^X_s \vee \F^Y_s \right)= \mathbb{P} \left( X_t=F \mid \sigma(X_s) \vee \F^Y_s \right), \quad \mathbb{P}\text{-a.s.}.
	\end{align}
\end{definition}
Here, the operator $``\vee"$ stands for the sigma-algebra generated by two sigma-fields $\F^X_{\cdot} $ and $ \F^Y_\cdot $.
\begin{definition}
	\label{defindefaultintensity}
	A family of $\Fb^Y$-adapted processes $\Lamb=(\Lambda_{EF}(t))_{t \ge 0}$  or $\lamb=(\lambda_{EF}(t))_{t \ge 0}$ is called the   default  intensity  family of an
	$\mathbb{N}$-valued    process  $X=(X_t)_{t\geq 0}$, if  for any  $F \in \mathbb{N}$,   the process $ X_F = (X_F(t))_{t \ge 0} $ is an $\check{\Fb}$-martingale, where
	\begin{align}
	X_F(t) := \mathds{1}_{F}(X_t)- \sum_{E\subseteq F} \int^t_0 \mathds{1}_{\{X_s = E\}} \dd \Lambda_{EF}(s), \; \text{ with } \;  \Lambda_{EF}(t) := \int_{0}^{t} \lambda_{EF}(s) \dd s,
\end{align}
 and  $\check{\Fb} = (\check{\F}_t)_{t \ge 0}:= (\F_t^X \vee \F_t^Y )_{t \ge 0}$.
\end{definition}

Similar to the bottom-up and top-down approaches, the intensity family $\Lamb$ ir $\lamb$ in our framework plays an important role in the compensator of the default process $X$, with $\lambda_{EF}(t)$ representing  the conditional default rate at time $t$ when obligors in set $E$ have already defaulted.
Notice that the condition (i) in Assumption \ref{assumption_default} is equivalent to
\begin{align}
\lambda_{EF}(t) = 0, \text{ whenever } F \neq E\cup\{i\} \text{ and } i\in E^c \text{ (complement of set $E$)}.
\end{align}

To make our framework more applicable in empirical studies, we assume the existence of  the intensity family $\Lamb$ (or $\lamb$) and the macroeconomic factor process $Y$ first, not the default process $X$  directly. The motivation follows from the fact  that one could apply the market credit derivative prices or spreads to recover the default intensity or contagion rate. We refer interested readers to   \cite{cont2010}, \cite{rama2013}, \cite{nickerson2017}, and the references therein, for related studies.
However,  it is rather difficult and technical to show the existence of an $\Fb^Y$-conditional Markov chain $X$ for a given intensity family
$\Lamb$, which is a main subject of the next section.

\section{{Characterizations} of the Default Process}
\label{sectionDefualtprocess}

In this section, we provide characterizations for the default process $X$, introduced in the previous section, in three steps.
First, we formulate the conditions that guarantee the existence of an $\Fb^Y$-conditional Markov chain $X$, which is used to model the default process in our framework,
see Assumptions \ref{assumption_Poisson} and \ref{assumption_intensity} and Theorem \ref{theorem_existence} in Section \ref{subsec_existence}.
Second, we derive the dynamics of $X$ through computing its conditional probabilities and expectations.
The key results under the general default intensity family are obtained in Theorem \ref{maintheoremforcaluclation}, which is essential in pricing and hedging problems of credit derivatives.
Third, we specify a class of processes for the intensity family in Assumption \ref{assumption_intensity_model} and simplify the results of Theorem \ref{maintheoremforcaluclation} to more tractable forms in Corollaries \ref{theoremparticular} and \ref{coroOneobligor}.

\subsection{The Existence of  the Default Process}
\label{subsec_existence}

In this subsection, we characterize the existence of an $\Fb^Y$-conditional Markov chain $X$ under Assumptions \ref{assumption_Poisson} and \ref{assumption_intensity}, and present results in Theorem \ref{theorem_existence}.
As mentioned in the introduction, such a Markov chain $X$ will be used to model the default process in our framework.

To construct an $\Fb^Y$-conditional Markov chain $X$, we begin with a given exogenous $\mathbb{R}^d$-valued stochastic process $Y$, which captures the macroeconomic information (or factors).
In this section, to obtain generality, we do not specify the dynamics of $Y$.
We consider the pricing problems in the next section when $Y$ is modeled by an affine jump-diffusion process.
In addition, we are given a family of stochastic processes $\Mb$ which are described below.

\begin{assum}
	\label{assumption_Poisson}
	$\Mb = (M_{EF}(t))_{t\geq 0}$ is a family of Poisson processes with intensity equal to one, where $E,F  \in \Nb$ and  $ E \subset F$.
	Furthermore, the Poisson family $\Mb$ and the macroeconomy process $Y$ are mutually independent.
\end{assum}

Next, we summarize the conditions that the intensity family $\Lamb$ (or $\lamb$) should satisfy, and assume those conditions hold in the rest of the paper.

\begin{assum}
	\label{assumption_intensity}
	The intensity family $\Lamb=(\Lambda_{EF}(t))_{t \ge 0}$, where $E,F  \in \Nb$, is a family of $\Fb^Y$-adapted processes, which satisfy the following conditions for all $t\geq 0$:
	\begin{itemize}
		\item[(A1)] $\Lambda_{EF}(t)=0$ (or $\lambda_{EF}(t)=0$), if $E \neq F$ or $F\neq E\cup \{i\}$, where $ i\in E^c $.
		\item[(A2)] $\Lambda_{EE}(t) = - \sum\limits_{E\neq F} \Lambda_{EF}(t)$ $\left(\text{or } \lambda_{EE}(t) = - \sum\limits_{E\neq F} \lambda_{EF}(t) \right)$.
		For notation simplicity, let $\Lambda_E(t):=-\Lambda_{EE}(t)$ and $\lambda_E(t):=-\lambda_{EE}(t)$.
		\item[(A3)] $\Lambda_{EF}(t)$ is  an increasing function of $t$, with $\Lambda_{EF}(0)= 0$.
		\item[(A4)] $\lim\limits_{t \to +\infty} \Lambda_{EF}(t) = +\infty$  for all  $F= E\cup \{i\}$ and $i\in E^c.$
	\end{itemize}
\end{assum}

\begin{remark}
	The essential part of Assumption \ref{assumption_intensity} is (A1), and it is equivalent to condition (i) in Assumption \ref{assumption_default}. (A2) is the Markov transition density requirement. (A3) and (A4) are essential to impose  positivity assumption on the family  $\lamb := (\lambda_{EF}(t))_{t\geq 0 }$.
	We do not make any further assumptions on the structure of  the intensity family $\Lamb$ or $\lamb$, which shall give our framework more flexibility and versatility to capture  default contagion.
\end{remark}

We end this subsection by the following crucial theorem which addresses the existence question of the default process $X$ when the exogenous process $Y$ and the intensity family $\Lamb$ are given.

\begin{theorem}
	\label{theorem_existence}
	Suppose an exogenous process $Y$ is given, and two families of processes $\bm{M}$ and $\Lamb$ are specified by Assumptions \ref{assumption_Poisson} and \ref{assumption_intensity},
	then there exists an ${\mathbb{F}^Y}$-conditional Markov chain $X$ with intensity family $\Lamb$ and $X_0 = \emptyset$.
\end{theorem}

\begin{proof}
	The construction of such an ${\mathbb{F}^Y}$-conditional Markov chain $X$ is achieved by carefully counting the jump times of a family of non-homogenous Poisson processes  and by verifying the conditions in Definitions \ref{definiconditionalmarkov} and \ref{defindefaultintensity}.
	The existence of $X$ is obtained under very general assumptions, and our framework embraces a broad class of credit risk models.
	The detailed proof is separated into three parts, which are delayed to Appendix \ref{sectionconstructionmarkov}.
\end{proof}

\subsection{The Dynamics of the Default Process}
\label{subsec_dynamics}

In this subsection, using the results from Theorem \ref{theorem_existence}, we obtain the dynamics of the default process $X$ through deriving its conditional probabilities and expectations, see Theorem \ref{maintheoremforcaluclation}.

To proceed, we introduce some notations to facilitate the presentation of key results below.
For any   $E,F \in \Nb$ satisfying $E\subseteq F$ and $\mid F/E\mid =n$, denote by $\Pi(F/E)$ the set of all the permutations of  $F/E$.
Here, $F/E$ is set $F$ ``minus" set $E$, namely, $F/E = \{x: \, x \in F \text{ and } x \not\in E \}$.
For any  $\bm{\pi}=(\pi_1,\cdots,\pi_{n})\in \Pi(F/E)$, define the sequence of sets $(F^{\bm{\pi}}_k)_{k=0,1,\cdots,n}$ by
\begin{align}
\label{eqFPI}
F^{\bm{\pi}}_0 :=E  \quad \text{and} \quad \ F^{\bm{\pi}}_k := F^{\bm{\pi}}_{k-1}\cup\{\pi_k\}, \; \ k=1,2,\cdots,n.
\end{align}
The theorem below contains the key results regarding the dynamics of the default process $X$.

\begin{theorem}
	\label{maintheoremforcaluclation}
	Suppose the macroeconomy process $Y$ is given, and Assumptions \ref{assumption_Poisson} and \ref{assumption_intensity} hold.
	For any  $0 \le s \le t < +\infty$ and $ F \in {\mathbb  N}$, we have
	\begin{align}
	\label{eqmain1}
	\mathbb{P}\left[X_{t}=F\mid  {\cal F}^X_s \vee \mathcal{F}^Y_{t} \right]=\sum_{E\subseteq F}\mathds{1}_{\{X_s=E\}} \cdot G(s,t;E,F),
	\end{align}
	and further, for any bounded (nonnegative) ${\mathcal{F}^Y_{t}}$-measurable random variable  $\xi $,
	\begin{align}
	\mathbb{E} \left[\mathds{1}_{\{X_{t}=F\}} \cdot \xi \mid {\cal F}^X_s \vee \mathcal{F}^Y_{s}  \right]=\sum_{E\subseteq F}\mathds{1}_{\{X_s=E\}}
	\cdot \mathbb{E} \left[\xi \, G(s,t;E,F)\Big\vert  \mathcal{F}^Y_{s}\right], \label{eqmain2}
	\end{align}
	where
	\begin{align}
	\label{definitionFunctionG}
	G(s, t;E,F) &:= \begin{cases}
	H_0(s,t;E),  & \text{ if } E = F \\[1ex]
	\sum\limits_{\bm{\pi} \in \Pi(F/E)} H_{\mid F/E\mid}(s,t; F^{\bm{\pi}}_0,\cdots,F^{\bm{\pi}}_{\mid F/E\mid}), & \text{ if }  E \subset F
	\end{cases}  \\[1ex]
	H_0(s,t ; E) &:= e^{-\int_s^t \lambda_{E}(u) \dd u},  \\
	H_{k+1}(s,t ; F^{\bm{\pi}}_0,\cdots,F^{\bm{\pi}}_{k+1})
	&:=\int_s^t \lambda_{F^{\bm{\pi}}_k F^{\bm{\pi}}_{k+1}}(v) \cdot e^{-\int^t_v \lambda_{F^{\bm{\pi}}_{k+1}}(u) \dd u}
	\cdot H_k(s,v;F^{\bm{\pi}}_0,\cdots, F^{\bm{\pi}}_k) \dd v.
	\label{defFuntionH}
	\end{align}
	
\end{theorem}

\begin{proof}
	The proof of Theorem \ref{maintheoremforcaluclation} is postponed to Appendix \ref{appen_maintheoremforcaluclation}.
	Note that
	$\lambda_E(t) = - \lambda_{EE}(t) = \sum_{E \neq F} \lambda_{EF}(t)$ for all $E \in \Nb$.
\end{proof}
Theorem \ref{maintheoremforcaluclation} explicitly describes the conditional dynamics of $X$ under ${\cal F}^X \vee \mathcal{F}^Y$  once the default intensity $\lamb = (\lambda_{EF}(t))_{t\geq 0}$ is known. Given the dynamics of $X$  in Theorem \ref{maintheoremforcaluclation}, the loss process $L$ in (\ref{eqn_L_equivalent}) is completely characterized and  employed to price  credit derivatives in the next section. One interesting and important feature of Theorem \ref{maintheoremforcaluclation} is that the functional $G$ (serves as the transition kernel from set $E$ to set $F$) only involves the default intensity $\lamb$.
Notice that when the intensity family $\lamb$ is completely arbitrary, the calculation complexity of the conditional probabilities in \eqref{eqmain1} is tremendous, since permutations of $F/E$ are involved.  However, if we assign a more tractable structure to  $\lamb$, the computation of \eqref{eqmain1} will be simplified dramatically.
This gives us computational advantage when pricing and hedging credit derivatives, see examples in the next section.

In our modeling, there are $N$ obligors, where $N$ is a positive integer. When $N=1$ or 2, the results in Theorem \ref{maintheoremforcaluclation} can be to much simpler forms, as in the corollary below.

\begin{corollary}
	\label{corollary_dynamics}
		Suppose the macroeconomy process $Y$ is given, and Assumptions \ref{assumption_Poisson} and \ref{assumption_intensity} hold.
	If there is only one obligor and  $\lambda_{\emptyset \{1\}} = \lambda>0$, we have, for all $0 \le s \le t < +\infty$, that
	\begin{align}
	\mathbb{P} \left[X_{t} = \emptyset \mid \mathcal{F}^X_{s}\vee {\cal F}^Y_s \right] &= \mathds{1}_{\{X_s=\emptyset\}}  e^{-\lambda (t-s)}, \\
	\mathbb{P} \left[X_{t} = \{1\} \mid \mathcal{F}^X_{s}\vee {\cal F}^Y_s \right] &= \mathds{1}_{\{X_s=\emptyset\}}  \left(1 -  e^{-\lambda (t-s)} \right) + \mathds{1}_{\{X_s=\{1\}\}}.
	\end{align}
	If there are two obligors and  $\lambda_{EF} = \lambda >0$ for all $E\subset\{1,2\}$ and $F=E\cup\{i\}, i \in E^c$, we have, for all $0 \le s \le t < +\infty$, that
	\begin{align}
	\mathbb{P} \left[X_{t} = \emptyset \mid \F^X_s \vee \F^Y_s \right] &= \mathds{1}_{\{X_s=\emptyset\}} e^{-2\lambda (t-s)}, \\
	\mathbb{P} \left[X_{t} = \{i\} \mid \F^X_s \vee \F^Y_s \right] &=  \mathds{1}_{\{X_s=\emptyset\}} \left(e^{- \lambda (t-s)} - e^{- 2\lambda (t-s)}\right) + \mathds{1}_{\{X_s=\{i\}\}} e^{-\lambda (t-s)},  \ \ i = 1,2, \\
	\mathbb{P} \left[X_{t} = \{1,2\} \mid \F^X_s \vee \F^Y_s \right] &= \mathds{1}_{\{X_s=\emptyset\}} \left(1- e^{- \lambda (t-s)}\right) + \mathds{1}_{\{X_s=\{1 \text{ or }  2\}\}} \left(1-  e^{- 2\lambda (t-s)}\right) +\mathds{1}_{\{X_s=\{1,2\}\}}.
	\end{align}
\end{corollary}
When $N=2$ (the case of two obligors), we observe that, starting from $X_s=\{\emptyset\}$, the probability that $X_{t}$ hits state $\{1\}$ or $\{2\}$ achieves its maximum when $t-s = \ln(2)/\lambda$. Furthermore, both probabilities are increasing when $t-s \in [0,\ln(2)/\lambda]$ and decreasing when $t-s \in [\ln(2)/\lambda, +\infty)$.

\subsection{The Modeling of the Intensity Family}
\label{subsec_intensity}

In the previous subsection, we have obtained the dynamics of the default process $X$ in Theorem \ref{maintheoremforcaluclation}, which is fully determined by the default intensity family $\lamb$. When the number of obligors in the economy is small ($N=1,2$), we  have simplified the results of Theorem \ref{maintheoremforcaluclation} to more tractable forms in Corollary \ref{corollary_dynamics}. However, $N$ is usually large in many financial products, such as iTraxx and CDX, and that makes the computation of  conditional probabilities in \eqref{eqmain1} time consuming and inefficient.
One remedy is to model the intensity family $\lamb$ of the default process $X$ in specific forms, then develop a more applicable version of Theorem \ref{maintheoremforcaluclation}.
The rest of this subsection is then devoted to the modeling of the intensity family $\lamb$, taking the dependence of the macroeconomic factor and intergroup contagion   into account. We consider the following model for the intensity family $\lamb$, which satisfies all the conditions imposed in Assumption \ref{assumption_intensity}.

\begin{assum}
	\label{assumption_intensity_model}
	Let $(\beta_i)_{i \in \N}$ and $(\rho_{ij})_{i,j \in \N}$ be nonnegative constants, and $h(\cdot)$ be a positive real-valued function with $h(0)=1$.
	We define, for all $E \in \Nb$, that
	\begin{align}
	\label{eqn_Lc}
	\Lc_E(i) &:= \begin{cases}
	h(|E|) \cdot \sum\limits_{j \in E} \rho_{ji}, & \text{ if } E \neq \emptyset \\[1ex]
	\beta_i, & \text{ if } E = \emptyset
	\end{cases} \\
	\label{eqn_Lc_bar}
	\text{and} \quad \overline{\Lc}_E &:= \sum\limits_{i \in E^c} \Lc_E(i), \quad \text{ with } \quad \overline{\Lc}_\N := 0.
	\end{align}
	
	Let $\Phi(\cdot, \cdot)$ be a  positive functional mapping from $[0,\infty) \times \mathbb{R}^d$ to $\mathbb{R}$.
	The intensity family $\lamb=(\lambda_{EF}(t))_{t \ge 0}$, where $E,F \in \Nb$, of the default process $X$ is given by
	\begin{align}
	\label{eqn_intensity}
	\lambda_{EF}(t) = \begin{cases}
	\Phi(t, Y_t) \cdot \Lc_E(i), & \text{ if } F= E \cup \{i\} \text{ and } i \in E^c \\
	- \Phi(t, Y_t) \cdot \overline{\Lc}_E, & \text{ if } E = F \\
	0, & \text{ otherwise}
	\end{cases},
	\end{align}
	where $(Y_t)_{t \ge 0}$ is the macroeconomy process, and $\Lc_E(\cdot)$ and $\overline{\Lc}_E$ are defined above by \eqref{eqn_Lc} and \eqref{eqn_Lc_bar}.
\end{assum}

\begin{remark}
	In Assumption \ref{assumption_intensity_model}, the function $\Phi(t, Y_t)$ describes the default rate inherited from  the macroeconomic factors, such as economic downturn, business cycles and financial crises.
	The function $\Lc_{E}(i)$ describes the intergroup default contagion rate induced by the defaulted obligors in $E$ on the survivor $O_i$, and $\overline{\Lc}_E$ is the aggregate  impact of defaulted obligors in $E$ on all survivors in $E^c$; while $\rho_{ji}$ is the individual contagion rate from $j$ to $i$.
	The function $h$ measures the magnitude of intergroup contagion.
	Hence, the model of $\lamb$, given by \eqref{eqn_intensity}, is able to capture the impact of macroeconomic factors and intergroup contagion on the default intensity.
	We end this remark by noting that the model \eqref{eqn_intensity} embraces a broad class of default contagion models, such as  the homogeneous contagion model and the near neighbor contagion model, see \cite{herb08}.
\end{remark}

The intergroup contagion magnitude function $h$ and the contagion rate matrix $(\rho_{ij})$, in the definition of $\Lc_{E}(i)$, may have opposite effects on the severity of the total default risk.
For instance, \cite{jorinZhang2007} find that, in the CDS market, credit events  could have contagion and competition effects, i.e., there are good and bad credit contagions.
By adopting their findings, the function $h$ may have positive ($h>1$) or negative ($0<h<1$) effect on credit spreads.
In the numerical studies of Section \ref{sectoionNumeric}, we take $h(n) = e^{-\delta n}$, where $\delta$ is calibrated from the CDO tranche quotes.
Our findings show that $\delta$ is  positive for both 5Y and 7Y CDX.NA.HY, which imply the credit events in CDX.NA.HY group have a overall competition (negative) effect.

Recall that $\Pi(F/E)$ contains all the permutations of $F/E$.
Take $E \subset F \in \Nb$ with $|F/E|=n$, where $n$ is a positive integer no greater than $N$, and select any $\bm{\pi}=(\pi_1,\cdots,\pi_n) \in \Pi(F/E)$, define
\begin{align}
\label{eqn_Lc_hat}
\widehat{\Lc}^{\bm{\pi}}(n) := \prod^{n -1}_{k=0} \Lc_{F^{\bm{\pi}}_k}(\pi_{k+1}),
\end{align}
where $F^{\bm{\pi}}_\cdot$ and $\Lc_\cdot(\cdot)$ are defined by \eqref{eqFPI} and \eqref{eqn_Lc}, respectively.

Let $l_0,l_1,\cdots,l_n$ be $n+1$ \emph{different} real numbers, where $n$ is a positive integer.
For all $m=1,2,\cdots, n$ and $ i=0,1,\cdots,m-1$,  we define
\begin{align}
\label{eqn_alpha}
\alpha^{(m)}_i(l_0,\cdots,l_m) := \frac{\alpha^{(m-1)}_i(l_0,\cdots,l_{m-1})}{l_m - l_i} \quad
\text{and} \quad   \alpha^{(m)}_m(l_0,\cdots,l_m) := -\sum_{i=0}^{m-1}\alpha^{(m)}_i(l_0,\cdots,l_m), \quad 
\end{align}
with $\alpha^{(0)}_0(l_0):=1$.

When the intensity family $\lamb$ is given by \eqref{eqn_intensity}, we simplify the results of Theorem \ref{maintheoremforcaluclation} in the corollary below, which is the key result of this subsection.

\begin{corollary}
	\label{theoremparticular}
	Suppose the macroeconomy process $Y$ is given, and Assumptions \ref{assumption_Poisson} and \ref{assumption_intensity_model} hold.
	Choose any $E,F \in \Nb$ with $E \subseteq F$ and $|F \backslash E| = n$.
	If $\overline{\Lc}_{{\F}^{\bm{\pi}}_i}\neq \overline{\Lc}_{\F^{\bm{\pi}}_j}$  whenever $i \neq j $ and $\bm{\pi} \in \Pi(F/E)$,
	then given $\{X_s = E\}$, we have for all $0\le s \le t <+\infty$ that
	\begin{align}
	\mathbb{P} \Big[X_t=F \Big| {\cal F}^X_s\vee {\cal F}^Y_t\Big] = \begin{cases}
	\exp\left(- \overline{\Lc}_E \cdot \int^t_s \Phi(u,Y_u) \dd u\right), & \text{ if } E = F \\
	\sum\limits_{\bm{\pi} \in \Pi(F/E)} \widehat{\Lc}^{\bm{\pi}}(n) \sum\limits^{n}_{i=0} \alpha^{(n)}_i(\bm{\pi}) \cdot \exp\left({- \overline{\Lc}_{F^{\bm{\pi}}_i} \cdot \int^t_s \Phi(u,Y_u) \dd u }\right),
	& \text{ if } E \subset F
	\end{cases}
	\end{align}
	and further, for any nonnegative  (or bounded) ${\cal F}^Y_t$-measurable random variable $\xi$,
	\begin{align}
	\mathbb{E} \Big[\mathds{1}_{\{X_t=F \}} \cdot \xi \Big| {\cal F}^X_s\vee {\cal F}^Y_s \Big] = \begin{cases}
	\mathbb{E} \left[ \xi \cdot e^{- \overline{\Lc}_E \cdot \int^t_s \Phi(u,Y_u) \dd u} \Big| \F_s^Y \right], & \text{ if } E = F \\
	\sum\limits_{\bm{\pi} \in \Pi(F/E)} \widehat{\Lc}^{\bm{\pi}}(n) \sum\limits^{n}_{i=0} \alpha^{(n)}_i(\bm{\pi}) \cdot
	\mathbb{E} \left[ \xi \cdot e^{- \overline{\Lc}_{F^{\bm{\pi}}_i} \cdot \int^t_s \Phi(u,Y_u) \dd u } \Big| \F_s^Y \right] ,
	& \text{ if } E \subset F
	\end{cases}
	\end{align}
	where $\overline{\Lc}_\cdot$ and $\widehat{\Lc}$ are defined by \eqref{eqn_Lc_bar} and \eqref{eqn_Lc_hat}, and
	$\alpha^{(n)}_i(\bm{\pi}):=\alpha^{(n)}_i( \overline{\Lc}_{F^{\bm{\pi}}_0}, \overline{\Lc}_{F^{\bm{\pi}}_1},\cdots,\overline{\Lc}_{F^{\bm{\pi}}_{n}})$, with $\alpha^{(n)}_i$ given by \eqref{eqn_alpha}.
\end{corollary}

\begin{proof}
	Please refer to Appendix \ref{appen_theoremparticular} for the proof of Corollary \ref{theoremparticular}.
\end{proof}

\begin{remark}
	In Corollary \ref{theoremparticular}, the conditional expectations rely on \textbf{time independent} $\widehat{\Lc}^{\bm{\pi}}$ and $\alpha^{(n)}_i$, and only $\int^t_s \Phi(u,Y_u) \dd u$ is needed. In comparison, the general results in Theorem \ref{maintheoremforcaluclation} requires the computations of  $\int^t_h\lambda_{F^{\bm{\pi}}_{k+1}}(u) \dd u$
	for all $0 \le h \le s$ and permutations $\bm{\pi}$ of $F/E$.
\end{remark}

Similar to Corollary \ref{corollary_dynamics}, when there is only one obligor, we can further reduce the results in Corollary \ref{theoremparticular} to much simpler forms, which are summarized in the corollary below. The results obtained in Corollary  \ref{coroOneobligor} coincide with those in the classical intensity model, see, e.g. \cite{duffiepansingleton2000} and \cite{cpr2004}.

\begin{corollary}
	\label{coroOneobligor}
	Suppose the macroeconomy process $Y$ is given, and Assumptions \ref{assumption_Poisson} and \ref{assumption_intensity_model} hold.
	If there is only one obligor, and its default time is given by $\tau$,   the following results hold for all $0 \le s \le t < + \infty$:
	\begin{itemize}
		\item[\rm (i)] The conditional survival probability is given by
		\begin{align}
		\mathbb{P}\left[ \tau >t \big| {\cal F}^X_s\vee {\cal F}^Y_s\right] = \mathbb{P}\left[ X_t = \{\emptyset\} \vert {\cal F}^X_s\vee {\cal F}^Y_s\right]
		= \mathds{1}_{\{X_s = \emptyset\}} \cdot \mathbb{E} \left[ e^{-\beta_1 \cdot \int_{s}^{t} \Phi(u,Y_u)du} \Big| {\cal F}^Y_s\right],
		\end{align}
		where $\beta_1$ is a constant from the definition of $\Lc$ in \eqref{eqn_Lc} and $\Phi$ is from the intensity model \eqref{eqn_intensity}.
		
		\item[\rm (ii)] The conditional default probability  is given by
		\begin{align}
		\mathbb{P}\left[ \tau \le t \big| {\cal F}^X_s\vee {\cal F}^Y_s\right]
		= \mathds{1}_{\{X_s = \emptyset\}} \cdot \left( 1 -  \mathbb{E} \left[ e^{-\beta_1 \cdot \int_{s}^{t} \Phi(u,Y_u)du} \Big| {\cal F}^Y_s\right] \right) +  \mathds{1}_{\{X_s = \{1\}\}}.
		\end{align}
	\end{itemize}
\end{corollary}

\begin{proof}
	The results follow immediately by calculating that $\overline{\Lc}_{\emptyset} = \beta_1$, $\overline{\Lc}_{\{1\}} = 0$, $\overline{\Lc}_{F^{\bm{\pi}}_0} = \beta_1$, $\overline{\Lc}_{F^{\bm{\pi}}_1} = 0$, $\widehat{\Lc}^{\bm{\pi}}(1) = \beta_1$, $\alpha^{(1)}_0 = -1/\beta_1$ and $\alpha^{(1)}_1 = 1/\beta_1$.
\end{proof}

\begin{remark}
	From assertion $\rm (i)$ in Corollary \ref{coroOneobligor}, we obtain
	\begin{align*}
	e^{-rt} \cdot \mathbb{P} (\tau>t ) = \mathbb{E} \left[ \exp \left(- \int_{0}^{t} (\beta_1 \Phi(u,Y_u)+r) \dd u \right) \right].
	\end{align*}
	Hence, we could interpret $\beta_1 \Phi(t,Y_t)$ as the instantaneous default risk premium at time $t$, and price a defaultable bond with maturity $T$ as a risk-free bond under a risk-adjusted discount rate $\beta_1 \cdot \Phi+r$.
\end{remark}

\section{Credit Derivatives Pricing under Affine Jump Diffusion Intensity}
\label{sectionCDOpricing}

In this section, we apply the results in Corollary \ref{theoremparticular} to price synthetic CDOs and CDX indexes under our default contagion model.
The intensity family $\lamb$ is given by \eqref{eqn_intensity} under an affine jump diffusion model, as specified in Assumption \ref{assumption_jump_diffusion_intensity}.
The general results are obtained in Proposition \ref{proposition_spread}, and those under two special models in Propositions \ref{propExampleHomo} and \ref{propexamplenonhomo}.

\subsection{Affine Jump Diffusion Intensity}

In Section \ref{subsec_intensity}, we specify the structure of the intensity family $\lamb$ in Assumption \ref{assumption_intensity_model}, in which the function $\Phi$ and the macroeconomy process $Y$ are in general forms.
To obtain explicit pricing formulas of credit derivatives, we further make the following assumptions for $\Phi$ and $Y$.

\begin{assum}
	\label{assumption_jump_diffusion_intensity}
	Let Assumption \ref{assumption_intensity_model} hold and the intensity process $\lambda_{EF}$ be given by \eqref{eqn_intensity}.
	Take $\Phi(t, Y_t) \equiv Y_t$ for all $t \ge 0$ and suppose the dynamics of $Y$ are governed by the following one-dimensional affine jump diffusion process:
	\begin{align}
	\label{eqn_jump_diffusion}
	\dd Y_t = \kappa (\theta-Y_t) \dd t + \sigma\sqrt{Y_{t}} \dd W_t + \dd J_t, \quad \text{ with } Y_0 = y_0,
	\end{align}
	where $\kappa$, $\theta$ and $\sigma$ are (positive) constants, $W=(W_t)_{t \geq 0}$ is a standard Brownian motion,
	$J = (J_t)_{t\geq 0}$ is a compound Poisson process with jump times from a Poisson distribution with intensity $l$ and jump sizes exponentially distributed with mean $\mu$.
	The processes $W$ and $J$ are mutually independent.
\end{assum}

\begin{remark}
	In the above model \eqref{eqn_jump_diffusion}, the parameter $\kappa$ is the mean-reversion rate of $Y$ to the long-term level $\theta$.
	The jumps of the compound Poisson process $J$ capture the sudden credit deterioration  events, which may lead to defaults of the obligors in the economy.
	Such an affine jump diffusion model is considered in  \cite{duffiepansingleton2000} and \cite{DGT2009}.
	A special case is the no-jump ($l=0$) model of \cite{feller1951}, which is used by  \cite{cox1985} to model stochastic interest rates.
\end{remark}

The following lemma is needed when deriving the pricing formula of credit derivatives in the sequel, see, e.g., Appendix of  \cite{duffiegarleanu2001} and \cite{Mortensen2005}.

\begin{lemma}
	\label{lemma_conditional_expectation}
	Let Assumption \ref{assumption_jump_diffusion_intensity} hold.
	For any $g>0$ and $t \ge 0$, we have
	\begin{align}
	\label{EGLAM}
	\mathbb{E} \Big[e^{-g \int_{0}^{t}Y_s \dd s }\Big]=e^{A(g,0,t)+  y_0 \cdot B(g,0,t)},
	\end{align}
	where $A(g,0,t)$ and $B(g,0,t) $ are given by
	\begin{align}
	\label{calAB}
	B(g,0,t)&=\frac{1-e^{bt}}{c_1+d_1e^{bt}}, \\
	A(g,0,t)&=\frac{ \kappa  \theta \gamma}{g bc_1d_1}\ln\left(\frac{c_1 + d_1e^{bt}}{-\gamma/g}\right) + \frac{  \kappa \theta}{c_1}t +\frac{l(c_2 d_1 - c_1 d_2)}{b c_1 c_2 d_2}\ln\left(\frac{c_2+d_2e^{bt}}{c_2+d_2}\right) + \left(\frac{l}{c_2} - l \right) t,
	\end{align}
	with $\gamma = \sqrt{\kappa^2 + 2g\sigma^2}$,
	$c_1 = - \frac{\gamma + \kappa}{2g}$, $d_1 = c_1 + \frac{\kappa}{g}$, $c_2 = 1 - \frac{\mu}{c_1}$, $d_2 = \frac{d_1 + \mu}{c_1}$,
	and $b=d_1g +g\frac{\kappa c_1 - \sigma^2 }{\gamma}$.	
\end{lemma}

\subsection{Index and Tranche Spreads}
A synthetic collateralized debt obligation (CDO) is a portfolio consisting of $N$ single-name credit default swaps (CDS) on obligors, with default times $\nu_1,\nu_2,\cdots, \nu_N$ and recovery rates $R_1,R_2,\cdots, R_N$. It is standard to assume that all obligors have the same nominal values, as in the cases of  iTraxx and CDX. Without loss of generality, the nominal value is set to one.
Then the accumulated loss $L$ is given by \eqref{eqn_L_equivalent}.
Usually, $L$ is represented as a percentage of the total nominal value at time 0 (which is equal to $N$). In the following, with slight abuse of notations, we set
\begin{align}
\label{eqn_new_L}
L_t  = \frac{R_X(t)}{N}, \quad t \ge 0,
\end{align}
where $R_X(t) := \sum_{i \in X_t} (1 - R_i)$. Notice that, with the notation of $R_X(t)$, the numerator in \eqref{eqn_new_L} is the loss $L$ in \eqref{eqn_L_equivalent}.

A CDO is specified by the attachment points $0=p_0 <p_1<p_2<\cdots<p_K \le 1$ and the corresponding trances $[p_{i-1}, p_i]$, where  $i=1,2,\cdots, K$. An agreement on  tranche $i$ is a bilateral contract, in which the protection seller agrees to pay the protection buyer all credit losses occurred in the interval $[p_{i-1}, p_i]$. The payments  from the seller are then made at the corresponding default times before or at $T$.
In exchange for protection, the buyer pays a periodic premium fee proportional to the current outstanding value on tranche $i$ up to $T$, probably reduced by the occurred default losses. The accumulated  loss of tranche $i$   is defined by
\begin{align}
\label{eqn_Li}
L^{(i)}(X_t):={(L_t-p_{i-1})^{+}-(L_t-p_i)^{+}} = \left(\frac{R_X(t)}{N}-p_{i-1}\right)^{+}-\left(\frac{R_X(t)}{N}-p_i\right)^{+}.
\end{align}
Suppose the  premium fees from the buyer are paid at discrete times $0=t_0<t_1<t_2\cdots<t_m=T$, with time increment $\Delta_k:=t_k-t_{k-1}$ for $k=1,2,\cdots,m$, and  the risk-free interest rate is constant $r$. It is often assumed in practice that premiums are paid quarterly, i.e., $\Delta_k = 1/4$, which is exactly the parameter we choose in the numerical studies.
At the inception of an agreement, depending on different product structures, such as CDX and iTraxx to name a few, the buyer is usually required to pay an upfront fee.
In a typical CDO tranche $i$ with upfront rate
$u^{(i)}$ and swap spread $s^{(i)}$, the cash flows are as follows:
\begin{itemize}
	\item (Default Leg) The protection seller covers tranche losses $L^{(i)}(X_t)$, given by \eqref{eqn_Li}.
	
	\item (Premium Leg) The protection buyer pays $u^{(i)} \Delta p_i = u^{(i)}(p_i - p_{i-1})$ at inception and $s^{(i)}(\Delta p_i -L^{(i)}(X_{t_{k-1}}))\Delta_k$ at each payment time $t_k$, where $k=1,\cdots,m$.
\end{itemize}
Therefore, the value of the default leg must agree with that of the premium leg, which gives the spread $s^{(i)}$ of tranche $i$ by
\begin{align}
\label{eqn_spread}
s^{(i)}=\frac{\sum\limits^m_{k=1}e^{-rt_k}\mathbb{E}\left[L^{(i)}(X_{t_{k}})-
	L^{(i)}(X_{t_{k-1}})\right] - u^{(i)} \Delta p_i} {\sum\limits^m_{k=1} e^{-rt_{k}}\left(\Delta p_i -\mathbb{E}[L^{(i)}(X_{t_{k-1}})]\right) \Delta_k },
\end{align}
where $\mathbb{E}$ denotes the expectation taken under the risk neutral probability measure $\mathbb{P}$.
Note that in  \eqref{eqn_spread} we have assumed that defaults only occur at the start point of each time interval $[t_{i-1},t_{i}]$.
In the literature, the end point and the middle point of each time interval $[t_{i-1},t_{i}]$ are both used as well, see, e.g. \cite{Mortensen2005}, \cite{errais2010} and \cite{rama2013}.
However, these different assumptions will only slightly affect our calibration results, but will not lead to different conclusions.
As a direct consequence of Corollary  \ref{theoremparticular} and formulas \eqref{EGLAM} and \eqref{eqn_Li}, the tranche spread $s^{(i)}$ is further computed in the proposition below.

\begin{proposition}
	\label{proposition_spread}
	Suppose the macroeconomy process $Y$ is given, and Assumptions \ref{assumption_Poisson} and \ref{assumption_jump_diffusion_intensity} hold.
	For the above specified CDO, we have the following results.
	\begin{itemize}
		\item[\rm (i)] The spread $s^{(i)}$  of tranche $i$ is given  by \eqref{eqn_spread}, and $ \mathbb{E}[L^{(i)}(X_{t_{k}})]$ is computed by
		\begin{align}
		\label{eqcdogxt}
		\mathbb{E}[L^{(i)}(X_{t_{k}})]&=\sum_{n=0}^N  \; \sum_{F \in \mathcal{A}(n)}  L^{(i)}(F) \sum_{\bm{\pi} \in \Pi(F/\emptyset)} \widehat{\Lc}^{\bm{\pi}}(n)\sum_{i=0}^n
		\alpha^{(n)}_i(\bm{\pi}) \cdot \mathbb{E}\left[e^{- \overline{\Lc}_{F^{\bm{\pi}}_i} \cdot \int_0^{t_k} \Phi(u, Y_u) \dd u}  \right],
		\end{align}
		where $\mathcal{A}(n) := \{ F \in \Nb: \,  |F| =n \}$, $L^{(i)}$ is given by \eqref{eqn_Li}, $\alpha^{(n)}_i(\bm{\pi}):=\alpha^{(n)}_i ( \overline{\Lc}_{F^{\bm{\pi}}_0}, \overline{\Lc}_{F^{\bm{\pi}}_1},\cdots,\overline{\Lc}_{F^{\bm{\pi}}_n})$ is given by \eqref{eqn_alpha}, operators $\overline{\Lc}$ and $\widehat{\Lc}$ are defined by \eqref{eqn_Lc_bar} and \eqref{eqn_Lc_hat}.

		\item[\rm (ii)] If all the obligors have the same recovery rate $R_i = R$,  the accumulated  loss $L^{(i)}(X_{t_{k}}) $ and its expected value are respectively given by
		\begin{align}
		\label{LIGN}
		L^{(i)}(X_{t_{k}}) &= I^{(i)}(|X_{t_{k}}|), \\
		\mathbb{E}[L^{(i)}(X_{t_{k}})]&=\sum_{n=\parallel V_{i-1} \parallel+1}^N  I^{(i)}(n) \sum_{F \in \mathcal{A}(n)}  \sum_{\bm{\pi}\in \Pi(F/\emptyset)}
		\widehat{\Lc}^{\bm{\pi}}(n)\sum_{i=0}^n\alpha^{(n)}_i(\bm{\pi })    \cdot \mathbb{E}\left[e^{- \overline{\Lc}_{F^{\bm{\pi}}_i} \cdot  \int_0^{t_k} \Phi(u, Y_u) \dd u}\right],
		\end{align}
		where 
		\begin{align} 
		\label{eqn_I}
		I^{(i)}(n) := \frac{1-R}{N} \left(\left( n - V_{i-1} \right)^{+} -\left( n - V_i \right)^{+}\right), \quad V_i := \dfrac{N}{1-R} \cdot p_i,
		\end{align}
		and $ \parallel V_i \parallel$ is the integer part of  $V_i$.
	\end{itemize}
\end{proposition}

\begin{proof}
	It follows immediately from Corollary \ref{theoremparticular}. The computations are omitted here.
\end{proof}

\subsection{Pricing Examples}
\label{examplessubsec}

In this subsection, we simplify the results of Proposition \ref{proposition_spread} under two specific contagion models:
(1) homogeneous contagion model (HCM) in Proposition \ref{propExampleHomo} and (2) near neighbor contagion model (NCM) in Proposition \ref{propexamplenonhomo}. 
Both models are special cases of the intensity model in Assumption \ref{assumption_intensity_model}.

\begin{assum}[Homogeneous Contagion Model]
	\label{exampleHOmo}
	Let the intensity family $\lamb$ be given by Assumption \ref{assumption_intensity_model}.
	We further assume that the default contagion is homogeneous among obligors.
	To be specific, we assume $\rho_{ij}=\rho$ for all  $i\neq j$,  $\Phi(t,y)=y$ and $h(n)=e^{-\delta n}$, where $\delta $ is a constant and can be interpreted as the contagion recovery rate.
\end{assum}

The homogeneous contagion model (HCM) introduced in Assumption \ref{exampleHOmo} is  different from that in \cite{frey2010dynamic} and \cite{laurent2011hedging}, where the authors  specify their default intensity $\lamb$ only through a homogeneous Poisson process, while our family of default intensities $\lamb$ consists of two components: one from macroeconomic factor $\Phi$ and the other from intergroup contagion matrix $(\rho_{ij})$.  Modeling all  heterogeneous intergroup contagion  leads to  a large scale of parameters and  numerical issues, such as computational inefficiency and loss of robustness in model parameters calibration. As a result, we will  mainly restrict to HCM in our numerical studies.
We remark that HCM is a reasonable assumption for CDO tranches on large indexes such as iTraxx and CDX, although this  may be an issue with equity tranches for which idiosyncratic risk is playing an important role.
Under the above HCM, we obtain the following results.

\begin{proposition}
	\label{propExampleHomo}
	Suppose the macroeconomy process $Y$ is given, and Assumptions \ref{assumption_Poisson},  \ref{assumption_jump_diffusion_intensity} and \ref{exampleHOmo} hold.	
	The spread $s^{(i)}$  of tranche $i$ is given by  \eqref{eqn_spread}, and the expectation  $ \mathbb{E}[L^{(i)}(X_{t_{k}})]$ is computed by
	\begin{align}
	\label{eqcdohom1}
	\mathbb{E} \left[L^{(i)}(X_{t_{k}}) \right]=\sum_{i=0}^{N-1}\Gamma_i \cdot \exp \big(A(a_i,0,{t_k})+ y_0 B(a_i,0,t_k) \big) + 1,
	\end{align}
	where functionals $A,B$ are from Lemma \ref{lemma_conditional_expectation},
	$a_0 :=\sum_{i=1}^N\beta_i$, $a_k :=\rho k(N-k)e^{-\delta k}$ for $k=1,\cdots,N$,
	and $\Gamma_i$ is defined, for all $i=0,1,\cdots, N-1$, by
	\begin{align}
	\Gamma_i := a_0\sum_{n=\max(i, \parallel V_{i-1} \parallel +1)}^N \frac{(n-1)!(N-1)!}{(N-n)!}\rho^{n-1} \, I^{(i)}(n) \, e^{-\delta n(n-1)/2} \, \alpha^{(n)}_i(a_0,a_1,\cdots,a_n).
	\end{align}
\end{proposition}

\begin{proof}
	Please refer to Appendix \ref{appen_propExampleHomo} for the proof.  	
\end{proof}

Next, we consider the near neighbor contagion model (NCM), where  each obligor $O_i$ can only impact its neighbors $O_{i-1}$ and $O_{i+1}$.
We describe the model in Assumption  \ref{examplenonhomo} and obtain related results in Proposition \ref{propexamplenonhomo}.

\begin{assum}[Near Neighbor Contagion Model]
	\label{examplenonhomo}
	Let the intensity family $\lamb$ be given by Assumption \ref{assumption_intensity_model}.
	We further assume that
	each obligor $O_i$ can only impact its neighbors $O_{i-1}$ and $O_{i+1}$. To be specific, we set $\Phi(t,y)=y$, $h(n)=e^{-\delta n}$ and
	\begin{align}
	\rho_{ji} = \begin{cases}
	\mathfrak{p}, & \text{ if } \{ i= j+1,   1\leq j\leq N-1\}    \text{ or }  \{j = N \text{ and } i = 1 \}\\
	\mathfrak{q}, & \text{ if } \{ i=  j-1, 2 \leq j \leq  N\} \text{ or  }  \{ j= 1 \text{ and } i=N \}\\
	0, & \text{ otherwise}
	\end{cases}
	\end{align}
	where $\mathfrak{p}>0$, $\mathfrak{q}>0$,  $\delta$ are  constants.
\end{assum}

\begin{proposition}
	\label{propexamplenonhomo}
	
	Suppose the macroeconomy process $Y$ is given, and Assumptions \ref{assumption_Poisson},  \ref{assumption_jump_diffusion_intensity} and  \ref{examplenonhomo} hold.	
	The spread $s^{(i)}$  of tranche $i$ is given by  \eqref{eqn_spread}, and the expectation  $ \mathbb{E}[L^{(i)}(X_{t_{k}})]$ is computed by	
	\begin{align}
	\label{nonhonexamtheorem}
	\mathbb{E}\left[L^{(i)}(X_{t_{k}})\right] =   1+ \overline{a}_0 \sum_{n=  \parallel V_{i-1} \parallel +1}^{N-1}  \overline{A}_n \overline{B}_n   + \overline{a}_0 \overline{A}_N \overline{B}_N ,
	\end{align}
	where $\overline{a}_0 :=\sum\limits_{i=1}^N\beta_i$, $\overline{a}_N=0$, $\overline{a}_k:=e^{-\delta k}(\mathfrak{p}+\mathfrak{q})$ for $k=1,2,\cdots, N-1$,  and
	\begin{align}
	\overline{A}_n &:= I^{(i)}(n) e^{-\delta n(n-1)/2}  (\mathfrak{p}+\mathfrak{q})^{n-1},  \\
	\overline{B}_n  &:= \sum_{i=0}^n\alpha^{(n)}_i(\overline{a}_0,\cdots, \overline{a}_n)  \exp\{A(\overline{a}_i,0,{t_k})+ y_0 B(\overline{a}_i,0,t_k)\}, \; n=1,\cdots,N-1,\\
	\overline{B}_N &:=  \sum_{i=0}^{N-1} \alpha^{(N)}_i(\overline{a}_0,\cdots, \overline{a}_N)   \exp\{A(\overline{a}_i,0,{t_k})+ y_0 B(\overline{a}_i,0,t_k)\},
	\end{align}
	with functionals $A,B$ as given in Lemma \ref{lemma_conditional_expectation}.
\end{proposition}

\begin{proof}
The proof is placed in Appendix \ref{appen_propexamplenonhomo}.
\end{proof}

\section{Numerical Studies}
\label{sectoionNumeric}

In Section \ref{subsec_sensitivity}, we demonstrate how to apply theoretical results from Propositions \ref{propExampleHomo} and \ref{propexamplenonhomo} to compute CDO index and tranche spreads, and calculate the attachment/detachment times in a toy example.
Furthermore, we carry out a sensitivity analysis to investigate the impact of various model parameters on tranche spreads.
In Section \ref{subsec_market}, we mainly focus on calibrating the homogeneous contagion model to market data and validate the practical use of our novel default risk model.

Throughout this section, the intensity family $\lamb$ and the macroeconomic factor $Y$ are specified as in Assumption \ref{assumption_jump_diffusion_intensity}.
We consider two specific default models in the studies, namely the homogeneous contagion model (HCM) from Assumption \ref{exampleHOmo} and the near neighbor contagion model (NCM) from Assumption \ref{examplenonhomo}.
We choose the following base parameters for our default contagion models.
\begin{itemize}
	\item The number of obligors is $N=125$ (e.g., the constitutes of iTraxx); the risk-free interest rate is $r=5\%$; the payments of CDS premiums are made quarterly, i.e., $\Delta_k \equiv 1/4 := \Delta$; the recovery rates are the same for all obligors, and $R = 40\%$.
	\item The parameters of the affine jump diffusion process of $Y$, given by \eqref{eqn_jump_diffusion}, are
	\begin{align}
	\kappa=0.6, \; \theta =0.02,\; \sigma=0.141, \; l=0.2, \; \text{and} \; \mu = 0.1
	\end{align}
	which are taken from Table 2 of  \cite{duffiegarleanu2001}.
	\item In the HCM, we set $\rho = 0.05$, $\delta = -0.008$, and $a_0 = 0.35$, see Assumption \ref{exampleHOmo} and Proposition \ref{propExampleHomo} for the meanings of these parameters.
	\item In the NCM, we set $\mathfrak{p}= 0.3$, $\mathfrak{q}=0.3$, $\delta = -0.7$, and $\bar{a}_0=0.35$, see Assumption \ref{examplenonhomo} and Proposition \ref{propexamplenonhomo} for the meanings of these parameters.
\end{itemize}

The programming codes for all numerical studies in this section are  written in Python and performed on a Thinkpad PC with Intel(R) 2.50 GHz processor and 8.0GB  RAM.

\subsection{Sensitivity Analysis}
\label{subsec_sensitivity}

We first compute the 5-year CDO tranche spreads using the formulas from Proposition \ref{propExampleHomo} (for HCM) and Proposition \ref{propexamplenonhomo} (for NCM), and list them in bp (1bp = $10^{-4}$) in Table \ref{tablespread}.
In this example, we take the same attachment points as those of iTraxx Europe, and the upfront rates are set from 500 bp to 0 bp, from equity tranche  to super senior tranche.
The upfront rates are artificial, only for illustrative purpose. For example, the equity spread of the iTraxx Europe is the upfront premium on the tranche nominal, quoted in percentage, plus the running fee 500 bp.
In the example, since $\Delta = 0.25$ and the CDO has 5 years maturity, the number of payments is $m = 20$.
The CPU running time for the HCM and NCM are respectively 5.73 seconds and 2.41 seconds.
\begin{table}[h]
	\centering
	\caption{5-year CDO Tranche Spreads under HCM and NCM}	
	\label{tablespread}
	\begin{tabular}{cccc}\hline
		Tranches	& Upfront Rate(bp)	&	HCM. Spread(bp)	&	NCM. Spread(bp)\\ \hline
		$[0,3\%]$	    &	500 &  1002	&	418	\\
		$[3\%,6\%]$	    &	400	&	840 &   190 \\
		$[6\%, 9\%]$	&	300	&	795 &   211	\\
		$[9\%, 12\%]$	&	200	&	777 &   235	\\
		$[12\%,22\%]$	&	100	&	739 &   259	\\
		$[22\%,60\%]$	&	  0	&	619 &   283	\\
		\hline
	\end{tabular}
	
\end{table}

Next, we calculate the attachment and detachment times (in years) under HCM.
From the results in Table \ref{tableattachTime}, we observe that the equity tranche will endure losses in half a year and be wiped out in 1.25 years.
However, higher tranches $[12\%,22\%]$ and $[22\%,60\%]$ would take decades to be wiped out.

\begin{table}[h]
	\centering
	\caption{Attachment and Detachment Time under  HCM}
	\label{tableattachTime}
	\begin{tabular}{cccc}
		\hline
		Tranches	&	 Detachment Default Obligors	&	Attachment Time  (year)	&	Detachment Time (year)	\\
		$[0,3\%]$	&	6	&	0.5	&	1.25	\\
		$[3\%,6\%]$	&	13	&	1.25	&	1.75	\\
		$[6\%, 9\%]$	&	19	&	1.75	&	2.25	\\
		$[9\%, 12\%]$	&	25	&	2.25	&	3	\\
		$[12\%,22\%]$	&	46	&	3	&	11	\\
		$[22\%,60\%]$	&	125	&	11	&	294	\\
		\hline
	\end{tabular}
\end{table}

We then focus on the sensitivity analysis of various factors on the CDO tranche spreads under HCM.
In the sensitivity studies, we investigate only one factor each time, and keep other factors the same as in the setup of base parameters.
The factors we consider are default contagion rate $\rho$, contagion recovery rate $\delta$,
{ number of payments} $m$, default recovery rate $R$, macroeconomy mean-reversion rate $\kappa$ and volatility $\sigma$. For illustrative purpose, we choose upfront rate $u^{(i)}=0$ for all tranches.
According to the graphs in Figure \ref{pic_sensitivity}, we arrive at the following conclusions.
\begin{itemize}
	\item The CDO tranche spreads are very sensitive to all factors considered here, except for the macroeconomy volatility $\sigma$.
	\item Among all six factors considered, only the default contagion rate $\rho$ is positively related with respect to the tranche spreads, while the rest shows negative relation.
	\item The tranche spreads are extremely elastic to the default contagion rate $\rho$ and contagion recovery rate $\delta$. One can interpret $\delta$ as the government intervene or self recovery rate of the group. The equity tranche is less sensitive to $\delta$ comparing with other tranches.
\end{itemize}

\begin{figure}[h!]
	\centering
	\caption{Sensitivity Analysis of CDO Tranche Spreads under HCM}
	\subfigure{\includegraphics[width=8.3cm]{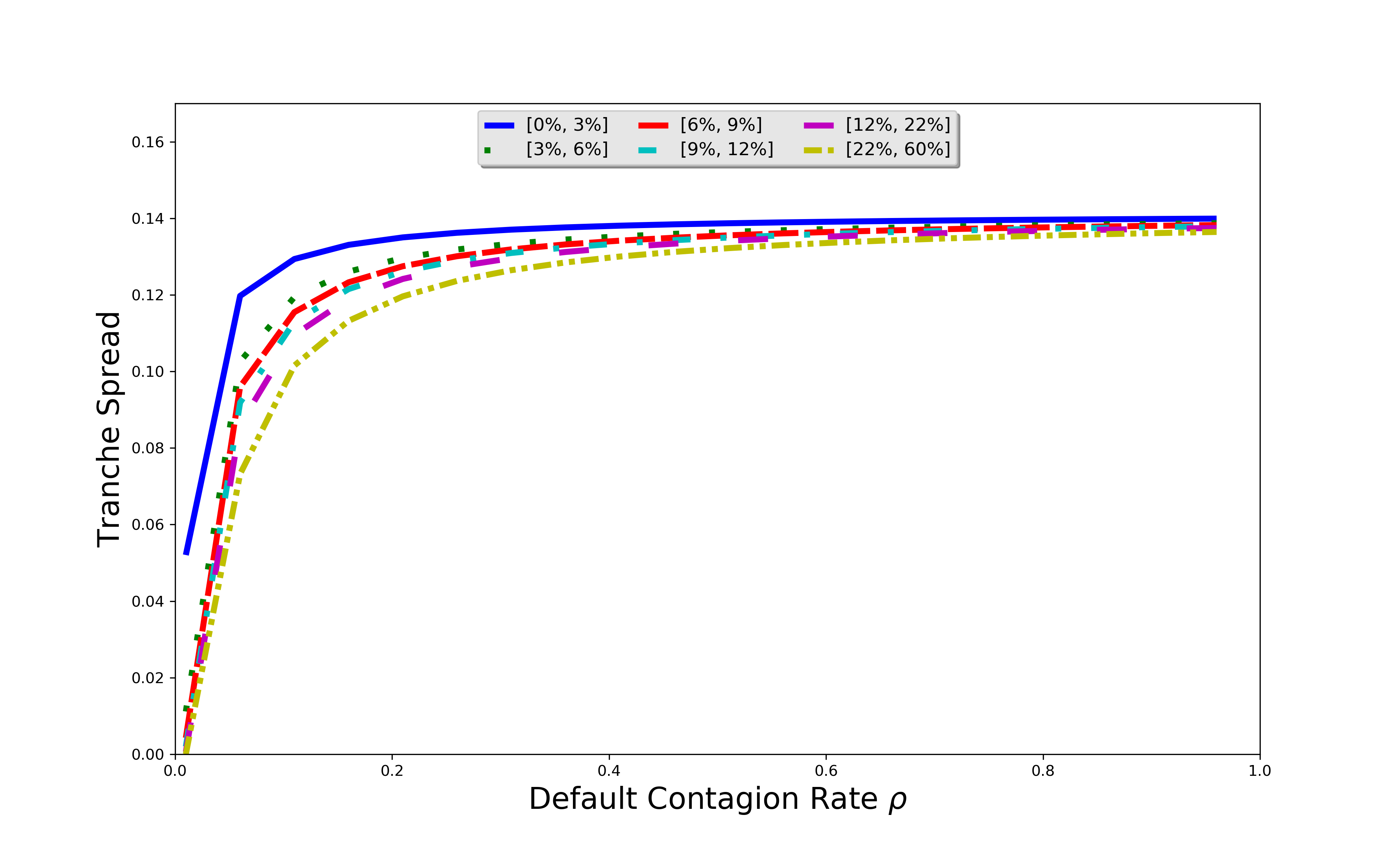}}
	\subfigure{\includegraphics[width=8.3cm]{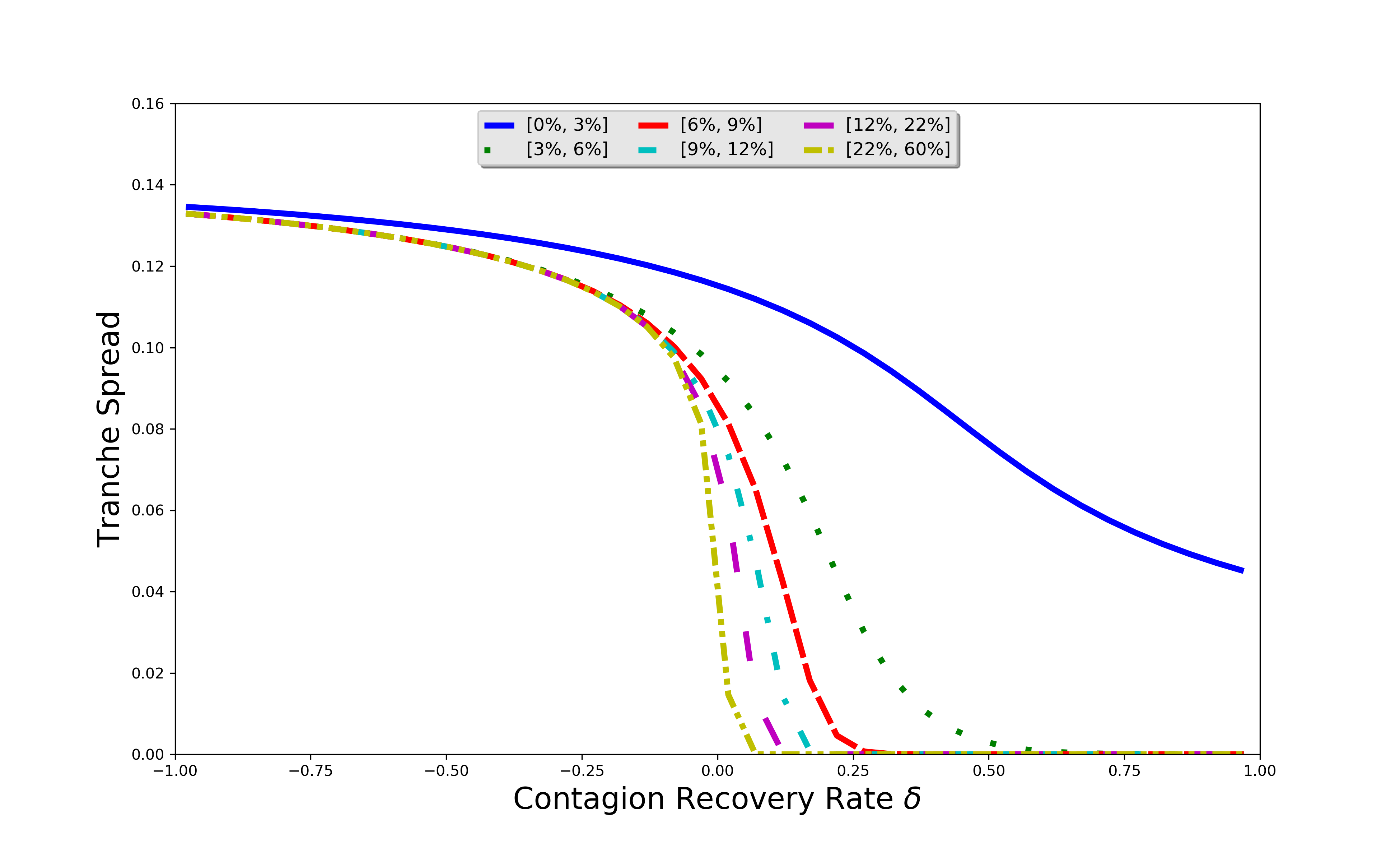}}
	\subfigure{\includegraphics[width=8.3cm]{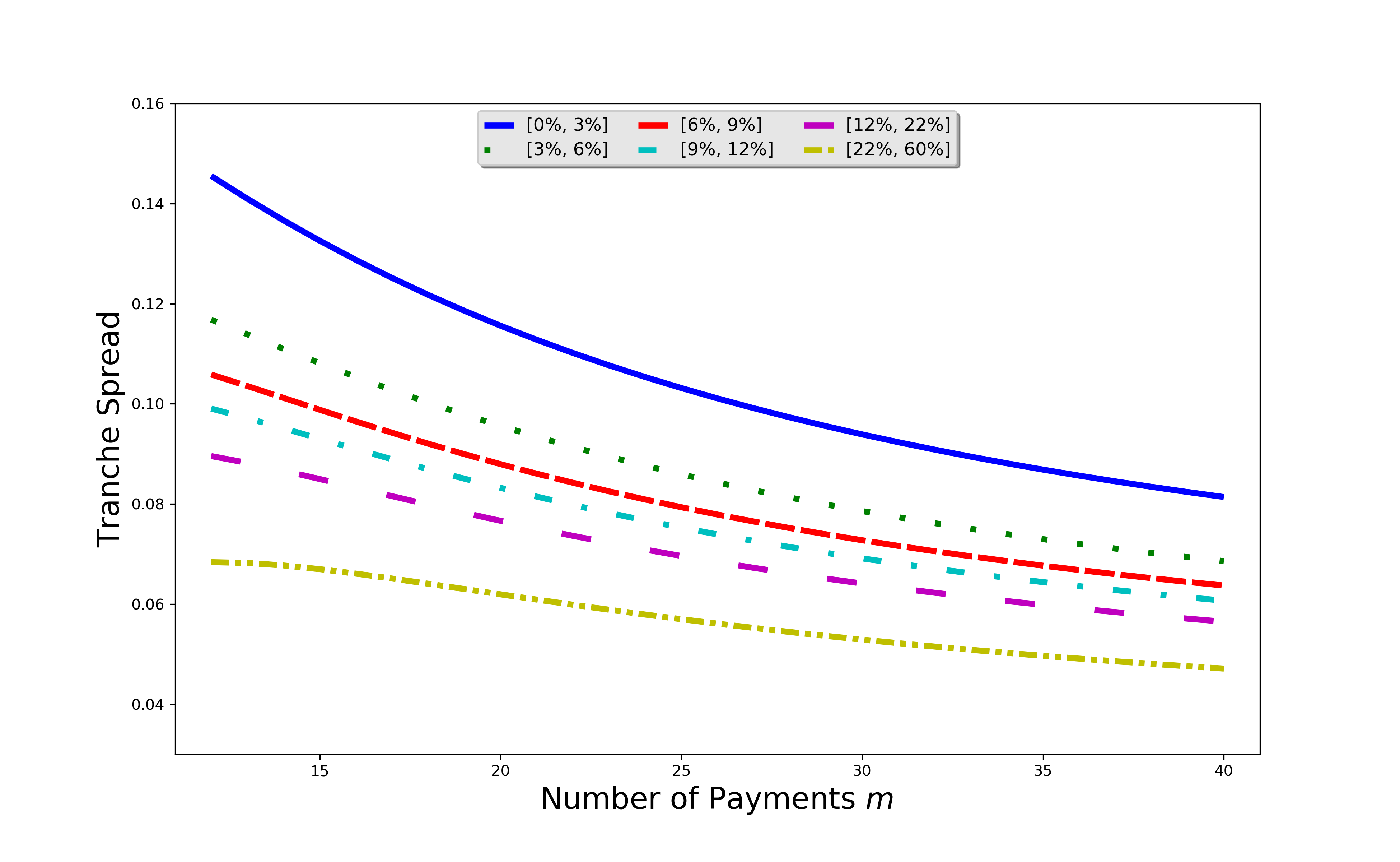}}
	\subfigure{\includegraphics[width=8.3cm]{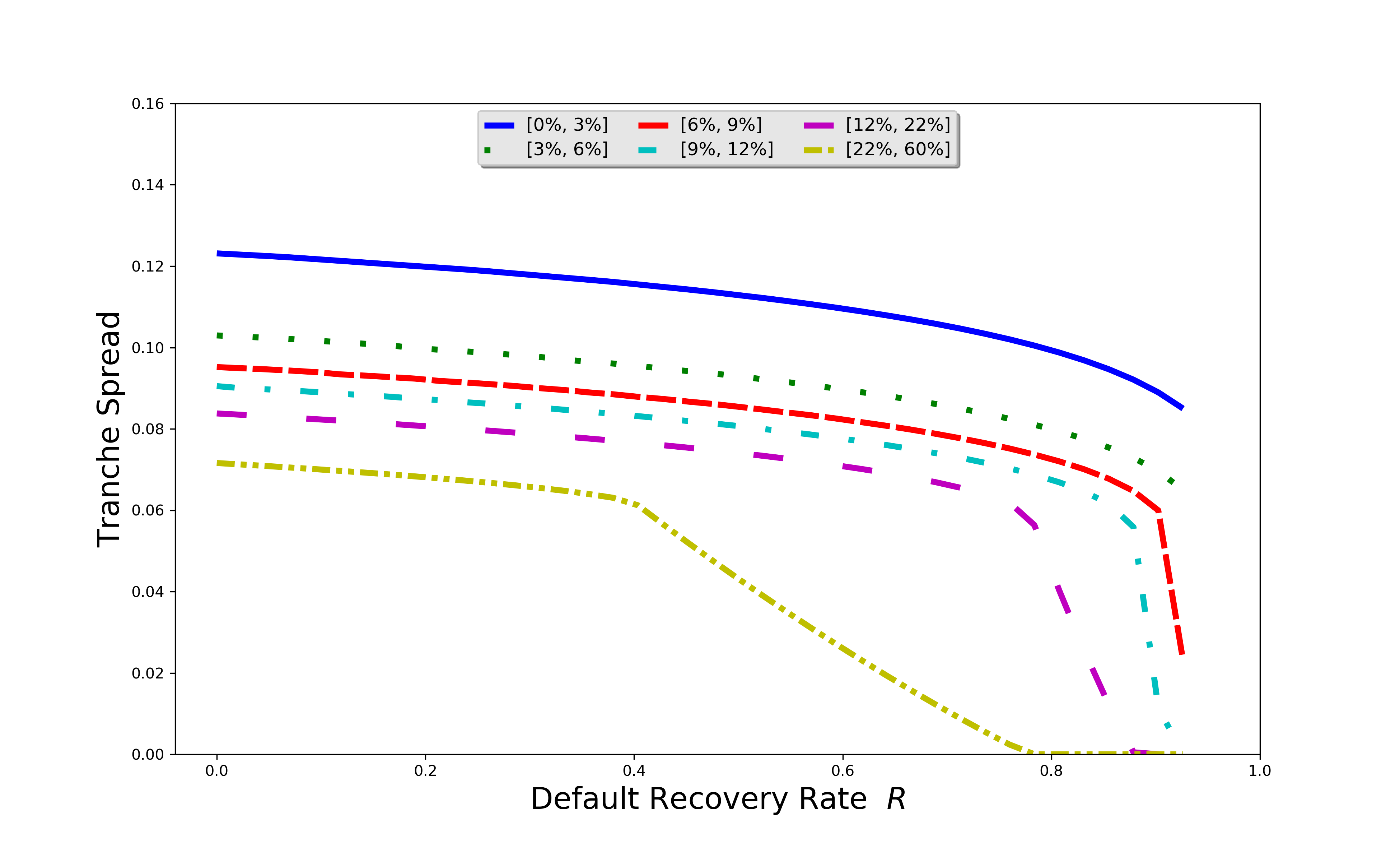}}
	\subfigure{\includegraphics[width=8.3cm]{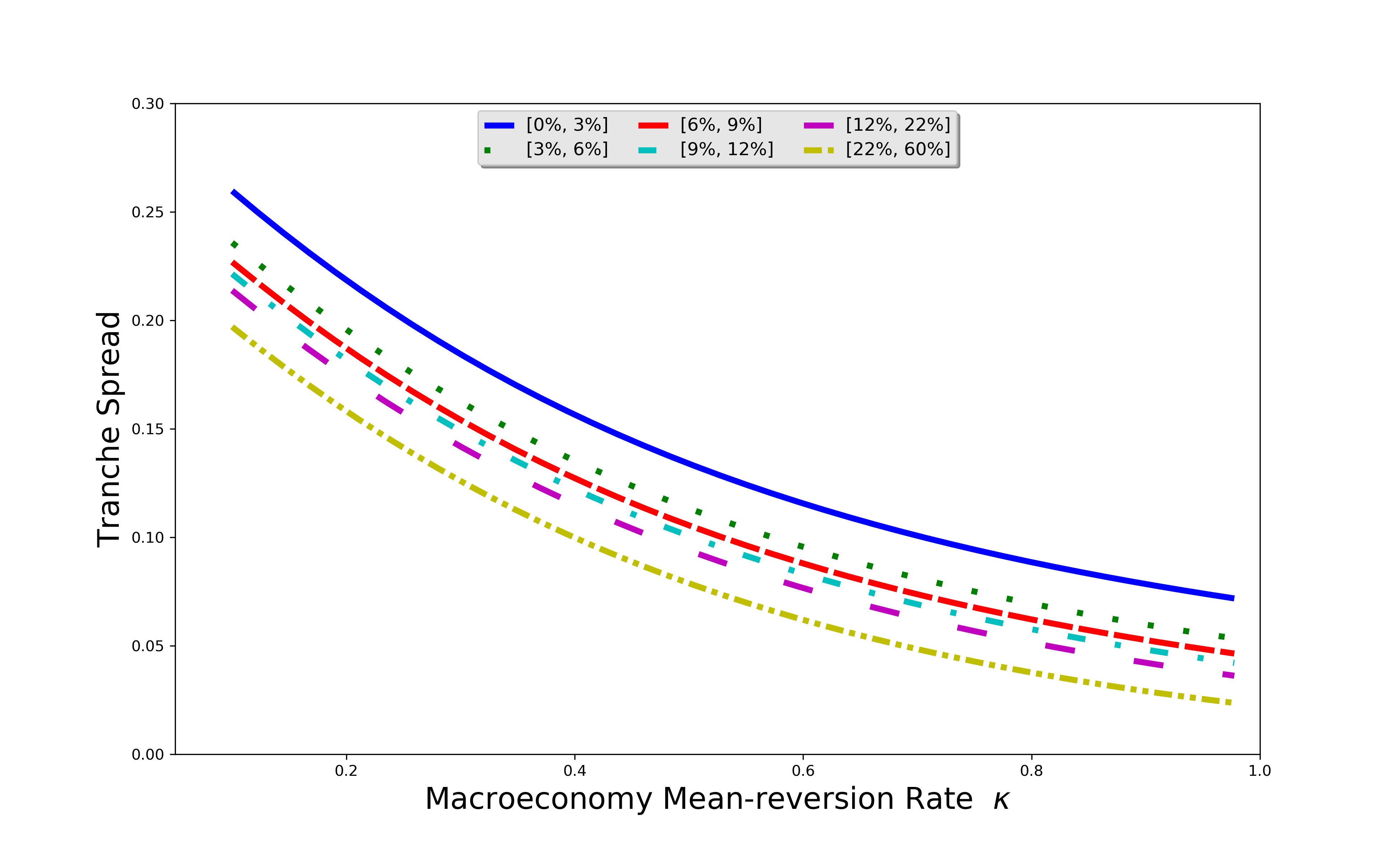}}
	\subfigure{\includegraphics[width=8.3cm]{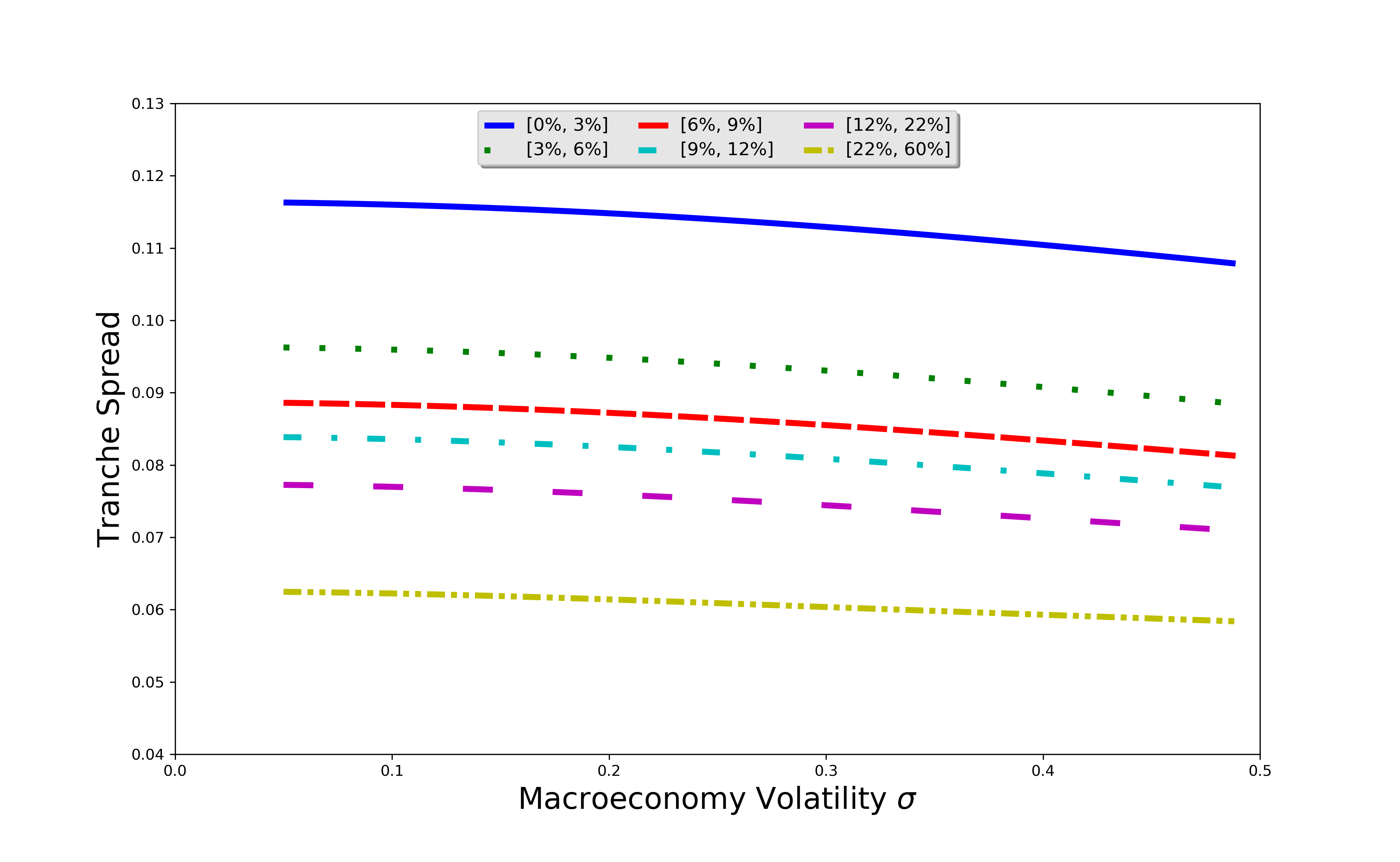}}
	\label{pic_sensitivity}
\end{figure}

Last, we study how the number of defaulted obligors evolves in time with respects to  four factors $\rho$, $\delta$, $\kappa$ and $\sigma$. The results are plotted in Figure \ref{pic_number}.
We find that the increase of the default contagion rate $\rho$ leads to the increase of defaulted obligors, since $\rho$ directly measures the default contagion rate.
In the meantime,  the number of defaulted obligors reduces when $\delta$, $\kappa$ or $\sigma$ increases since it alleviates the severity of contagions .

\begin{figure}[h!]
	\centering
	\caption{Evolution of the Number of Defaulted Obligors under HCM}
	
	\subfigure{\includegraphics[width=8cm]{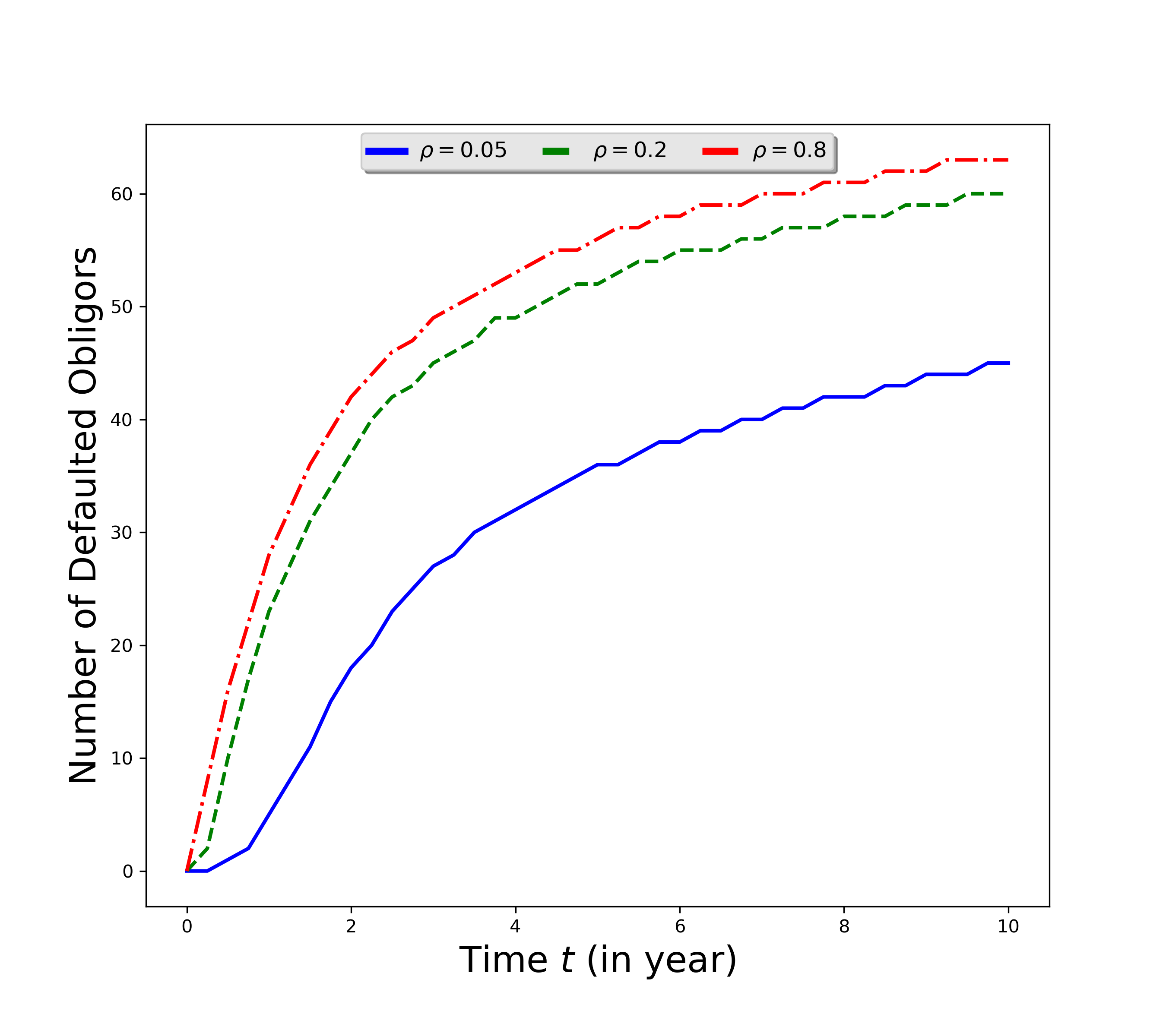}}
	\subfigure{\includegraphics[width=8cm]{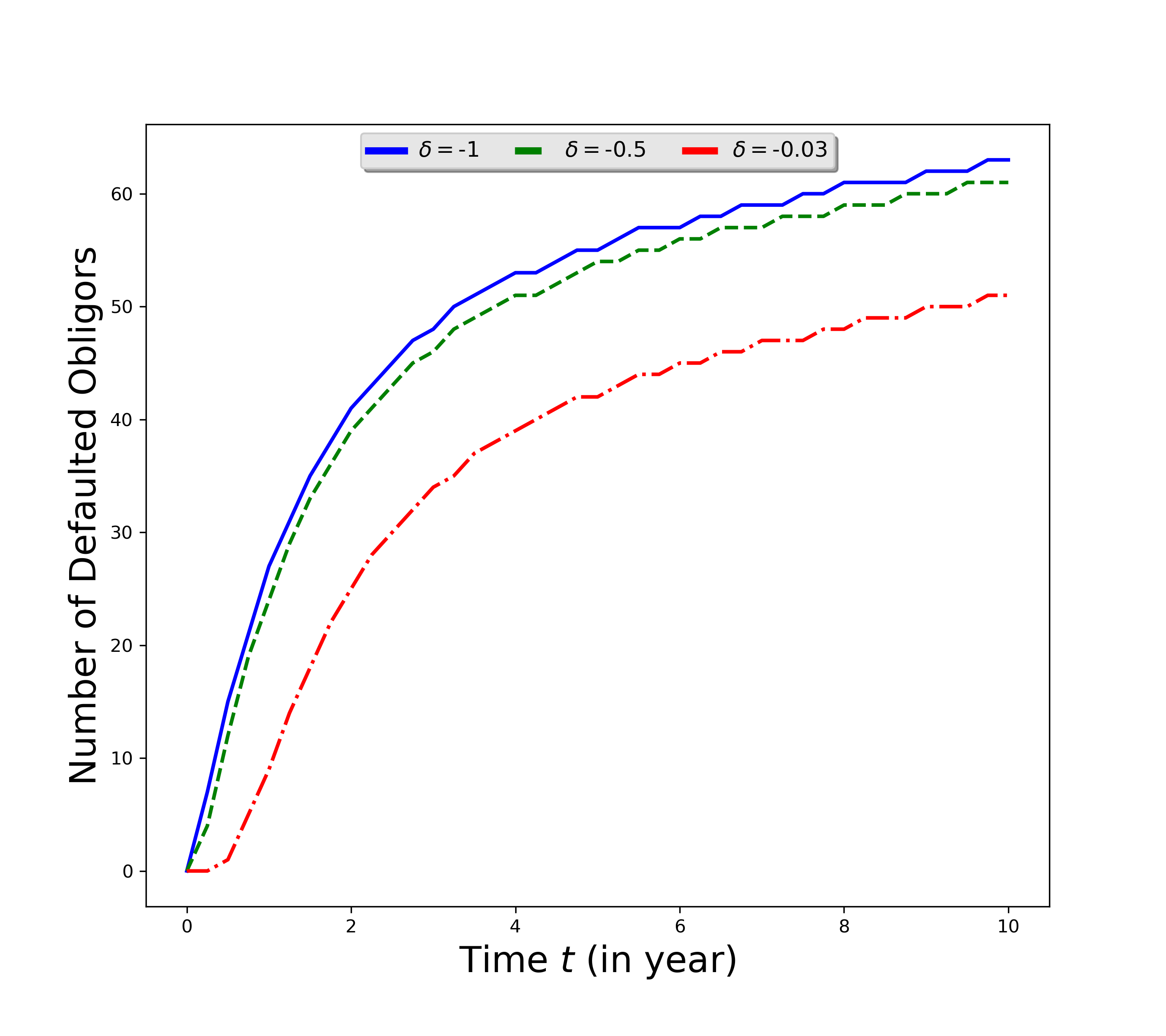}}
	\subfigure{\includegraphics[width=8cm]{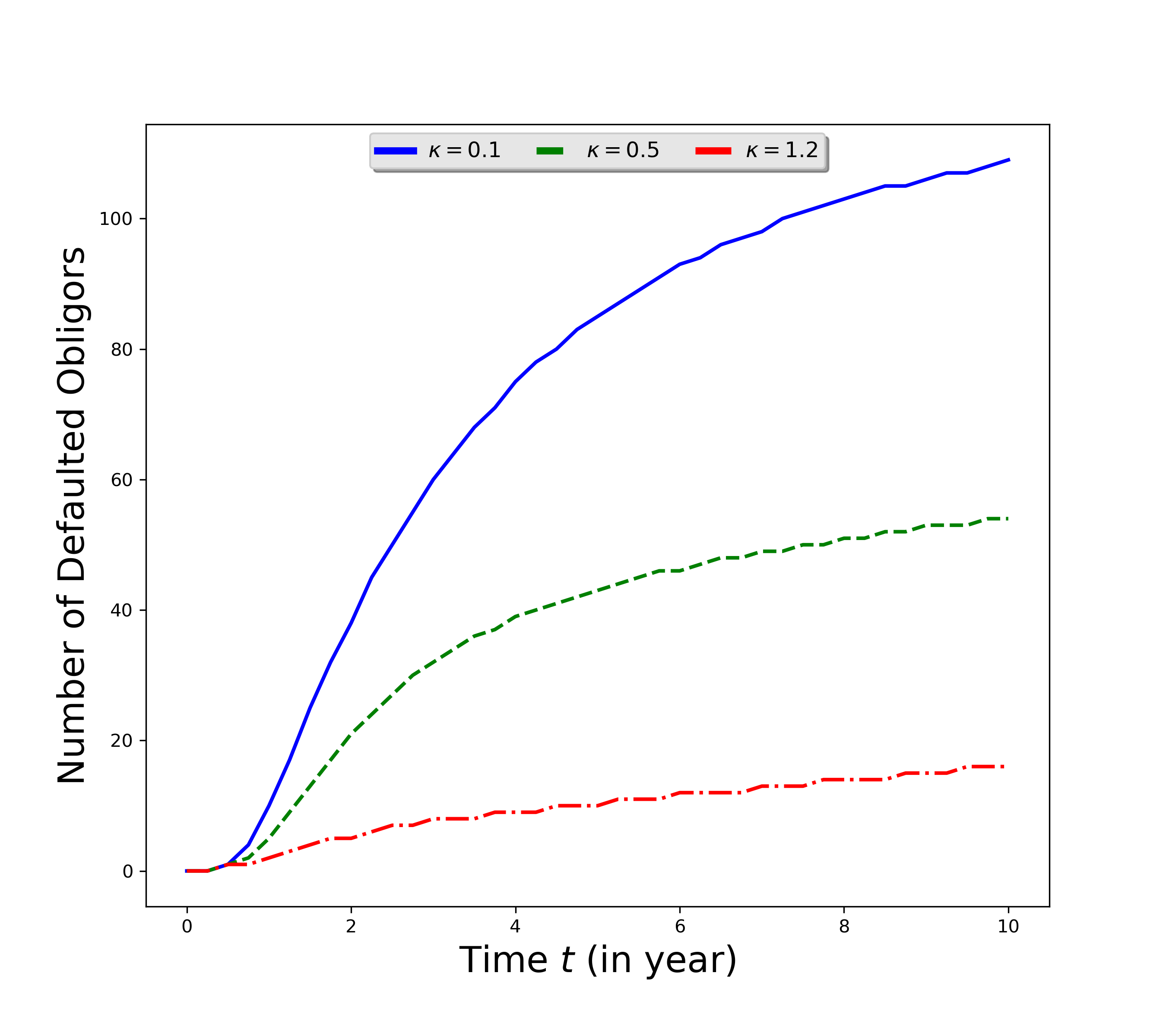}}
	\subfigure{\includegraphics[width=8cm]{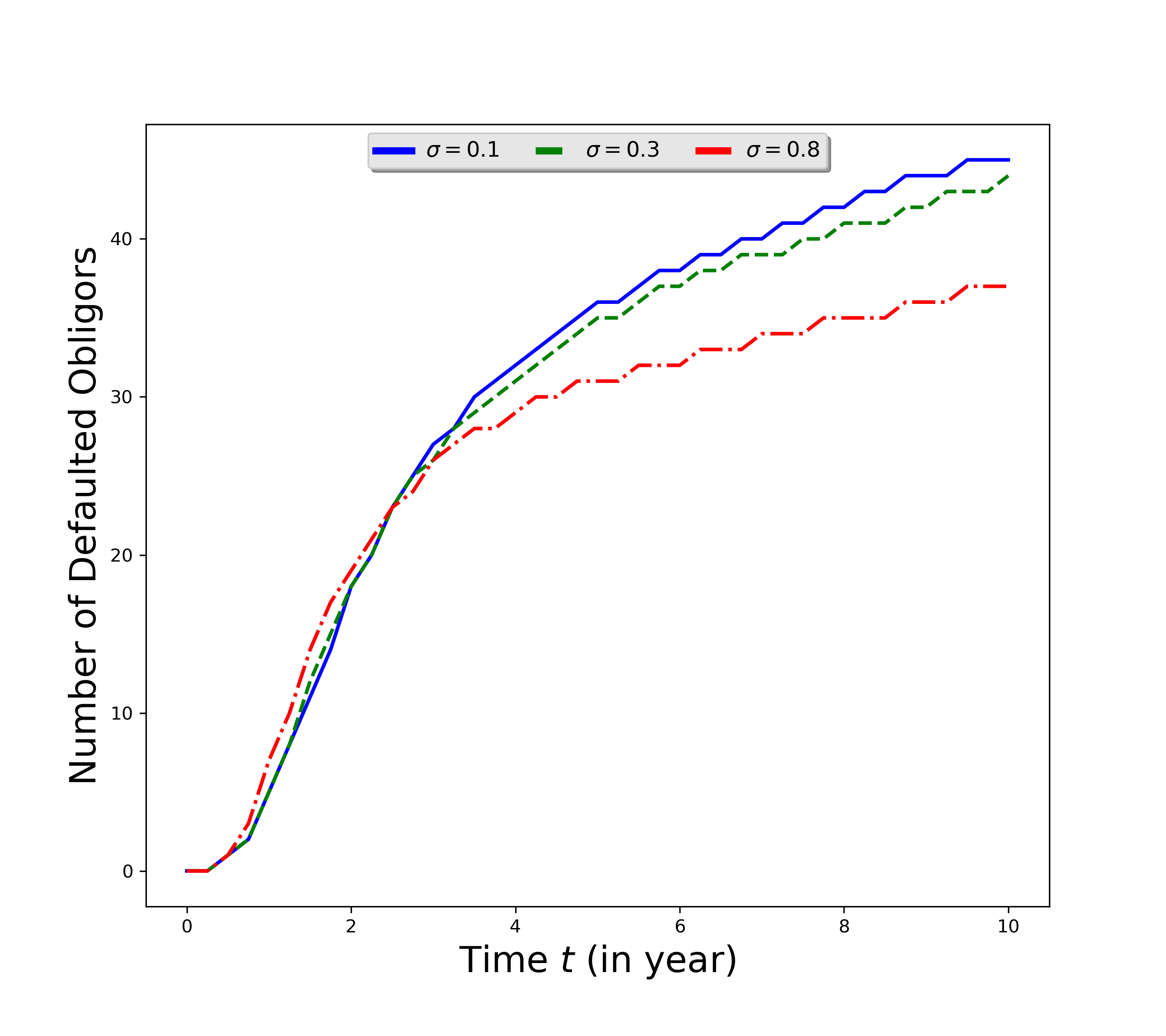}}
	\label{pic_number}
\end{figure}

\subsection{Market Calibration}
\label{subsec_market}

In Section \ref{subsec_sensitivity}, the base parameters are pre-selected, not calibrated using market data.
In this subsection, we use the market data\footnote{We use the data from  \cite{giesecke2007}.} of the CDX North American High Yield (CDX.NA.HY) index and its spreads observed on $5/11/2007$ to calibrate the parameters of the default intensity family $\Lamb$.
The CDX.NA.HY index constitutes of equally weighted $N=100$ obligors with attachments points $0,10,15,25,35$. The first two tranches $[0,15]$ and $[15,25]$ are quoted as a percentage of the upfront fee, while the others are quoted as a percentage of the running spread fee.

In the remaining studies, we only consider the homogeneous contagion model (HCM), introduced in Assumption \ref{exampleHOmo}, since the near neighbor contagion model (NCM) does not fit the market data well.
As a result, we aim to use the CDX.NA.HY market data to estimate the parameters vector $\bm{x} = (a_0, \rho, \delta, \kappa, \theta, \sigma, \mu,l,y_0)$.
To obtain the best fit $\hat{\bm{x}}$, we solve the following minimization problem
\begin{align}
\min_{ \bm{x} \in \Theta} \; \sum_{i=1}^{5} \left(\frac{{Trache \ i}^{model} - \overline{{Trache \ i}^{market}}}{\overline{{Trache\  i}^{market}}} \right)^2,
\end{align}
where $\Theta = (0, 2)\times(0,2)\times(-2, 1)\times(0,7)\times(0,7)\times(0,0.4)\times(0,5)\times(0,1)\times(0,10)$.\footnote{\cite{giesecke2010} consider similar feasible region in their studies. Changing the feasible region  $\Theta$ will only slightly affect calibration.}
${Trache \ i}^{model}$ is tranche $i$'s spread under HCM (see Proposition \ref{propExampleHomo}). $\overline{{Trache\  i}^{market}}$ is the average of tranche $i$'s bid-ask quote, where $i=1,2,3,4$, and $\overline{{Trache\  5}^{market}}$ is the CDX.NA.HY index.

We consider both 5Y and 7Y CDX.NA.HY indexes in the calibration.
We apply the sequential least squares programming (SLSQP) method built in  Python Scipy package to solve the above minimization problem.
The step size used for numerical approximation of the Jacobian is set to 1.49e-10, and the precision goal for the minimization value of the stopping criterion is 1e-15.
We first calibrate the model parameters to two indexes separately and list the results in Table \ref{tab_calibration}.

The calibrated optimal parameters for 5Y CDX.NA.HY are
$a_0 =1.135,  \, \rho = 0.00258, \, \delta=0.0149, \,  \kappa=0.958, \,  \theta=0.680, \, \sigma = 0.125,  \, \mu=2.380, \, l =0.236, \, y_0=0.998.$
The calibrated optimal parameters for 7Y CDX.NA.HY are
$a_0 = 1.199, \,\rho=0.00356, \,\delta=0.00950, \,\kappa=1.400,\, \theta=0.884,\, \sigma=0.382,\, \mu=0.362, \,l=0.320, \, y_0=1.000$
The average absolute percentage errors (AAPE) of the 5Y and 7Y CDX.NA.HY indexes are $4.36\%$ and $4.73\%$ respectively, both of which are on reasonable levels since liquidity risk and market makers' premium are included in the market prices.
We share the same view as \cite{Mortensen2005} that it is difficulty to rule out supply and demand effects caused by market segments or market inefficiency, and a prefect fit to the market prices should perhaps not be expected.

\begin{table}[h!]
	\centering
	\caption{Separate Calibration of 5Y and 7Y  CDX.NA.HY indexes}
	\label{tab_calibration}
	\begin{tabular}{ccccccc}
		\hline
		Tranche	&	5Y-Bid	&	5Y-Ask	&	5Y-Model	&	7Y-Bid	&	7Y-Ask	&	7Y-Model	\\
		\hline
		$[0,10\%]$	&	70.50\%	&	70.75\%	&	66.70\%	&	80.13\%	&	80.38\%	&	78.39\%	\\
		$[10\%,  15\%]$	&	34.25\%	&	34.50\%	&	32.89\%	&	55.50\%	&	55.75\%	&	53.44\%	\\
		$[15\%, 25\%]$	&	316.00	&	319.00&	337.72	&	582.00	&	587.00	&	626.24	\\
		$[25\%, 35\%]$	&	79.00	&	81.00	&	78.82	&	180.00	&	183	.00&	180.07	\\
		Index	&	262.85	&	263.10	&	248.00	&	307.50	&	307.75	&	278.43	\\
		MinObj	&		&		&	0.011	&		&		&	0.016	\\
		AAPE	&		&		&	4.36\%	&		&		&	4.73\%	\\
		\hline
	\end{tabular}
\end{table}

Next, we use the joint data of two indexes and redo the calibration. The results are obtained in Table \ref{table57jointly}. In the case of joint calibration, the optimal model parameters are
$a_0 = 1.0372, \, \rho=0.00558,\,  \delta = 0.0264, \, \kappa = 1.219, \, \theta = 0.898,
\, \sigma = 0.375, \, \mu=2.495, \, l=0.155, \, y_0=4.063$.

\begin{table}[ht]
	\centering
	\caption{Joint Calibration of 5Y and 7Y CDX.NA.HY indexes}
	\label{table57jointly}
	\begin{tabular}{ccccccc}
		\hline
		Tranche	&	5Y-Bid	&	5Y-Ask	&	5Y-Model	&	7Y-Bid	&	7Y-Ask	&	7Y-Model	\\
		\hline
		$[0,10\%]$	&	70.50\%	&	70.75\%	&	67.22\%	&	80.13\%	&	80.38\%	&	77.61\%	\\
		$[10\%,  15\%]$	&	34.25\%	&	34.50\%	&	33.72\%	&	55.50\%	&	55.75\%	&	54.26\%	\\
		$[15\%, 25\%]$	&	316.00	&	319.00	&	342.02	&	582.00	&	587.00&	604.84	\\
		$[25\%, 35\%]$	&	79.00	&	81.00	&	77.46	&	180.00	&	183.00	&	174.16	\\
		Index	&	262.85	&	263.10	&	245.87	&	307.50	&	307.75	&	273.66	\\
		MinObj	&		&		&	 	&		&		&	0.031	\\
		AAPE	&		&		&	 	&		&		&	4.83\%	\\
		\hline
	\end{tabular}
\end{table}

Under both top-down and bottom-up approaches, the contributions of  systematic  and idiosyncratic default risks are independent.  Precisely, the individual default intensity is assumed to take the form of $\overline{\lambda}^i(t) = \text{constant} \cdot \lambda(t) + \lambda^i(t)$, where $\lambda$ and $\label_i$ are independent processes representing the  systematic  and idiosyncratic components respectively, see \cite{Mortensen2005}.
Our numerical studies show that the systematic default risk coupled with default contagion risk (which is model by the individual contagion rate matrix $(\rho_{ij})_{i, j \in \N}$ under our framework) could have the leading component of the total default risk even   without individual  idiosyncratic factor.  
Such a key finding is consistent with the conclusions in \cite{jorinZhang2007}.
For single name CDS, individual  idiosyncratic risk might play a key  role; while for CDS index such as CDX and iTraxx, idiosyncratic risk is less significant in the aggregate default effect.

Last, we are concerned with an important parameter in HCM, the default contagion rate $\rho$, see Assumption \ref{exampleHOmo}.
Using the separately calibrated parameters (except $\rho$) for the 5Y and 7Y CDX.NA.HY indexes, we calculate the {\it implied default contagion rate} $\rho$ and present the results for the corresponding tranches in Table \ref{tablerhoimplied}.
Note that the implied $\rho$ of tranche $i$ is the one that solves the equation $s^{(i), Model }(\rho) =  s^{(i), Market}$, similar to the implied volatility of call/put options.
We observe a ``smile'' pattern of the implied default contagion rate $\rho$,
similar to the volatility ``smile" of options and the implied correlation ``smile'' in CDX tranches, see \cite{okaneLivesey2004}.
One possible explanation is that, the CDO tranches are segmented and each tranche contains a mixture of effects, including systematic and idiosyncratic credit risk,  liquidity effect, and supply and demand for certain tranches.

\begin{table}[h]
	\centering
	\caption{Implied Default Contagion Rate $\rho$}
	\label{tablerhoimplied}
	\begin{tabular}{ccc}
		\hline
		Tranche	&	5Y-Implied $\rho$	&	7Y-Implied $\rho$		\\
		\hline
		$[0,10\%]$	&	0.0027	&	0.010		\\
		$[10\%,  15\%]$	&	0.00092	&	0.0082		\\
		$[15\%, 25\%]$	&	0.0026&	0.0075		\\
		$[25\%, 35\%]$	&	0.0026	&	0.0072		\\
		Index	&	0.0027	&	0.0082		\\
		\hline
	\end{tabular}
\end{table}

\section{Conclusion}
\label{sectionConlusions}
We propose a novel default contagion framework on credit risk modeling, which takes into consideration the dynamical contagion effect among obligors and the impact of macroeconomic factors.
We consider a group of defaultable obligors and model the default process by a set-valued Markov process $X=(X_t)$, where $X_t$ is the set of all obligors that have defaulted by time $t$. 
We are able to derive the dynamics of the default process $X$ in explicit forms, and apply the results to  obtain analytic pricing formulas for credit debt obligations (CDOs).
The homogeneous contagion model (HCM) and near neighbor contagion model (NCM) are studied as special cases within our general framework.

In numerical studies, we demonstrate how analytic pricing results can be easily programed to compute the tranche spreads, and investigate the impact of various model factors on the tranche spreads.
Furthermore, we use the 5Y and 7Y CDX.NA.HY market data to calibrate the HCM and validate the practical applications of our new framework.
The model fits the 5Y and 7Y CDX.NA.HY tranche spreads and indexes reasonably well.
Our empirical findings support that systematic default risk coupled with default contagion among obligors could have the leading component of the total default risk, which is in line with the results of  \cite{jorinZhang2007}.

\newpage

\appendix
\section*{Appendix}
\section{Construction and Characterization of the Default Process $X$}
\label{sectionconstructionmarkov}

In this section, under Assumptions \ref{assumption_Poisson} and \ref{assumption_intensity}, we construct the default process $X$ through the pair $(\tau_n, X_{\tau_n})_{n=0,1,\cdots,N}$ in Appendix \ref{appen_construction}, and prove the conditional Markov property of $X$ in Appendix \ref{appen_markov} and the martingale property of $X$ in Appendix \ref{appen_martingale}, respectively.
Theorem \ref{theorem_existence} then follows immediately once the construction of $X$ is done, and the Markov and martingale properties are shown.

Recall $\N=\{1,2,\cdots,N\}$, where $N$ is the number of obligors in the group, and $\Nb$ is the $\sigma$-algebra of $\N$ consisting of all the subsets of $\N$.
To proceed, we make the following definitions:
\begin{align}
	{\mathbb{N}}^2_+ &:= \{(E,F): E,F\in {\mathbb{N}} \text{ and } E\subseteq F\}, \quad
	{\mathbb{N}}^2_{++} := \{(E,F)\in {\mathbb{N}}^2_+ : E\neq F\}, \\
	 E^+(i) &:= E\cup\{i\}, \quad E^{-}(i) := E /\{i\}, \quad \text{ for all }  i\in {\mathcal{N}} \text{ and } E\in {\mathbb{N}}.
\end{align}
Under a complete probability space $(\Omega,{\cal C},\mathbb{P})$,  an exogenous $\mathbb{R}^d$-valued stochastic process $Y$ is given. 
Suppose a family of Poisson processes $\Mb=\{(M_{EF}(t))_{t\geq 0} : (E,F)\in{\mathbb{N}}^2_{++} \}$ with intensity one is chosen according to Assumption \ref{assumption_Poisson}.
As an immediate consequence of Assumption \ref{assumption_Poisson}, the processes $Y$ and $M_{EF}$ are mutually independent for all $(E,F)\in{\mathbb{N}}^2_{++}$.
In addition, a family of processes $\Lamb:=(\Lambda_{EF}(t))_{t\geq 0}$ is given, which satisfies all the conditions imposed in Assumption \ref{assumption_intensity}.

For any $0 \le s \le t$ and $E,F \in \Nb$, let
$
\Lambda_{EF}(s,t):= \Lambda_{EF}(t)-\Lambda_{EF}(s) = \int_s^t \, \lambda_{EF}(u) \dd u,
$
and define the process {  $\Mhat_{EF} = (\Mhat_{EF}(t))_{t \ge 0}$} by
\begin{align}
\Mhat_{EF}(t) := M_{EF}(\Lambda_{EF}(t)), \quad \text{ with } (E,F)\in \Nb^2_{++},
\end{align}
and the $\sigma$-fields below
\begin{align}
 \F^{EF}_t &:= { \sigma \big\{ \Mhat_{EF}(s) : 0 \le s \le t \big\} \cup {\mathcal C} \text{-negligible sets} }, \\
\F^{EF}_{\infty} &:=  \bigvee_{t\geq 0}{\cal F}^{EF}_t, \quad \text{ and } \quad
\F^{(n)}_{\infty} := \bigvee_{(E,F)\in{\mathbb{N}}^2_{++}:\mid F\mid\leq n}{\cal F}^{EF}_{\infty}.
\end{align}
Here, the operator $ \bigvee_{i\in \text{index}} {\cal H}_i $ stands for the sigma-algebra generated by all indexed $({\cal H}_i)_{i\in \text{index}}$ (the index set could be uncountable).
Recall Assumption \ref{assumption_intensity}, if $F\neq  E^{+}(i)$ where  $i\in E^c$, then
$\Mhat_{EF} (t)=0 $ for all $t\geq 0$.
The proposition regarding $\Mhat_{EF}$ below  is straightforward to check, and hence the proof is omitted.

\begin{proposition}
\label{property_NEE}
Let Assumptions \ref{assumption_Poisson} and \ref{assumption_intensity} hold.
The process $\Mhat_{EF}$, where $(E,F)\in \Nb^2_{++}$, satisfies the following properties:
\begin{enumerate}
\item[\rm{(i)}] For any $E \in \Nb$, $i \in E^c$, integer $n \ge 0$, and $0 \le s < t$,
\begin{align}
\mathbb{P}\left[\Mhat_{EE^+(i)}(t) - \Mhat_{EE^+(i)}(s) = n \Big| \mathcal{F}^Y_{\infty}\bigvee{\cal F}^{EE^+(i)}_t\right]
&= \mathbb{P}\left[ \Mhat_{EE^+(i)}(t) - \Mhat_{EE^+(i)}(s) = n \Big| \mathcal{F}^Y_{\infty}\right] \\
&=\exp\left(-\Lambda_{EE^+(i)}(s,t)\right) \times \frac{(\Lambda_{EE^+(i)}(s,t))^n}{n!}.
\end{align}

\item[\rm{(ii)}] For any $E \in \Nb$, $i \in E^c$, integers $m,n \ge 0$, and $0 \le s < t <u$,
\begin{align}
&\; \mathbb{P} \left[\Mhat_{EE^+(i)}(u) - \Mhat_{EE^+(i)}(t)=m, \; \Mhat_{EE^+(i)}(t) - \Mhat_{EE^+(i)}(s) = n \Big| \mathcal{F}^Y_{\infty} \right] \\
 =\; &\exp\left( -\Lambda_{EE^+(i)}(s,u)\right) \times \frac{(\Lambda_{EE^+(i)}(t,u))^m}{m!} \times \frac{(\Lambda_{EE^+(i)}(s,t))^n}{n!}.
\end{align}

\item[\rm{(iii)}] For any $(E,F)$, $(E', F') \in \Nb^2_{++}$ and $(E,F)\neq (E',F')$, $\Mhat_{EF}$ and $\Mhat_{E'F'}$ are conditionally independent on $\F^Y_\infty$.
\end{enumerate}
\end{proposition}

Essentially, Proposition \ref{property_NEE} shows that, for any fixed $E\in \mathbb{N}$ and $i\in E^c$, the process $\Mhat_{EE^+(i)}$ is an $\mathcal{F}^Y_{\infty}$-conditional inhomogeneous Poisson process with intensity $\Lambda_{EE^+(i)}$.

\subsection{Construction of the Default Process $X$}
\label{appen_construction}
In this subsection, we construct the default process  $X$ by induction on the pair $(\tau_n, X_{\tau_n})_{n=0,1,\cdots,N}$.
Recall that under our framework, $\tau_n$ is the $n$-th default time and  $X_{\tau_n}$ is the set of obligors that have defaulted by time $\tau_n$.
Once $(\tau_n, X_{\tau_n})$ are constructed for all $n=0,1,\cdots,N$, we follow \eqref{eqn_L_equivalent} and define the default process $X$ by
\begin{align}
X_t := X_{\tau_n}, \quad \text{ if } \tau_n \le t < \tau_{n+1},
\end{align}
where $\tau_0 = 0$ and $\tau_{N+1} = + \infty$.
Note that $X_t = \N$ for all $t \ge \tau_N$.

The induction algorithm below allows us to construct a sequence for the pair $(\tau_n, X_{\tau_n})_{n=0,1,\cdots,N}$.
\begin{itemize}
\item[Step 1.] As convention, let $\tau_0 = 0$ and $X_{\tau_0} = 0$.

\item[Step 2.] Assume  $(\tau_n, X_{\tau_n})$ are defined for $n < N$ and satisfy  that
\begin{itemize}
	\item[(i)] Both $\tau_n $ and  $ X_{\tau_n}$ are ${\cal F}^{(n)}_{\infty}$-measurable, and  $\mathbb{P}[\tau_n < +\infty ]=1$.
	
	\item[(ii)]  $\mathbb{P}[ \mid X_{\tau_n} \mid =n]=1$.
\end{itemize}

\item[Step 3.] For any $ E\in {\mathbb{N}}$ with $ \mid E \mid = n <N$, and $i \in E^c$, define
\begin{align}
\tau_{n+1}(E,i) &:= \inf \left\{t>\tau_n: \Mhat_{EE^{+}(i)}(t) \neq \Mhat_{EE^{+}(i)}(\tau_n) \right\}, \\
\tau_{n+1} &:= \sum_{E \in \Nb:\, \mid E\mid=n}\mathds{1}_{\{X_{\tau_n}=E\}} \cdot \min \{\tau_{n+1}(E,i):i\in E^c\}, \\
X_{\tau_{n+1}} &:= E\cup \left\{i\in E^c: \tau_{n+1}=\tau_{n+1}(E,i) \right\},  \quad \text{given} \; \{X_{\tau_n}=E\}.
\end{align}

\end{itemize}

Intuitively, $\tau_{n+1}(E,i)$ is the default time of obligor $O_i$, given that obligors in set $E$ have already defaulted, where $i \in E^c$.
$\tau_{n+1}$ is the $(n+1)$-th default time, given that the default process at $\tau_n$ is $X_{\tau_n}$.
It is worth pointing out that  the set $\left\{i\in E^c: \tau_{n+1}=\tau_{n+1}(E,i) \right\}$ is not empty, since $\mathcal{N}$ is finite. The following lemma completes the definition of $(\tau_{n},X_{\tau_n})$.

\begin{lemma}
	Let Assumptions \ref{assumption_Poisson} and \ref{assumption_intensity} hold.
	For any integer $1\leq n<N,$ assume $(\tau_i, X_{\tau_i})_{i=0,1,\cdots,n}$ are defined as in Step 2 of the above algorithm,  the following two assertions hold:
	\begin{itemize}
	\item[\rm{(i)}] Both $\tau_{n+1}$ and $X_{\tau_{n+1}}$ are  ${\cal F}^{(n+1)}_{\infty}$-measurable.
	\item[\rm{(ii)}] $ \mathbb{P}[\tau_{n+1} < +\infty ]=1 $ and $\mathbb{P}[ | X_{\tau_{n+1}}| =n+1]=1$.
\end{itemize}
\end{lemma}
\begin{proof}
(i)  Since $\tau_n$ and $X_{\tau_n}$ are $ {\cal F}^{(n)}_{\infty}$-measurable,   $\tau_{n+1}$ is ${\cal F}^{(n+1)}_{\infty}$-measurable by construction.
For any $E \in \Nb$, we have
\begin{align}
\{X_{\tau_{n+1}}=E\} &= \bigcup_{i\in E}\{X_{\tau_n}=E^-(i), \, X_{\tau_{n+1}}=E \} \\
&=\bigcup_{i\in E} \{ X_{\tau_n}=E^-(i), \, \tau_{n+1}=\tau_{n+1}(E^-(i),i) \}\in{\cal F}^{(n+1)}_{\infty}.
\end{align}
Hence, we conclude $X_{\tau_{n+1}}$ is ${\cal F}^{(n+1)}_{\infty}$-measurable.

(ii) To show  $ \mathbb{P}[\tau_{n+1} < +\infty ]=1 $, it suffices to show that $\mathbb{P}[\tau_{n+1}(E,i)  = +\infty] = 0$ for all  $E \in \Nb$ and $i\in E^c$.
Since $\mathbb{P}[\tau_n<+\infty]=1$, we obtain
\begin{align}
\mathbb{P}[\tau_{n+1}(E,i)=+\infty] &= \lim_{t\rightarrow +\infty}\mathbb{P}[\tau_{n+1}(E,i)> t>\tau_{n}]\\
	&= \lim_{t\rightarrow +\infty} \mathbb{P}\left[ t>\tau_{n}, \, \Mhat_{EE^+(i)}(t) - \Mhat_{EE^+(i)}({\tau_n}) =0 \right] \\
	&\leq \lim_{t\rightarrow +\infty} \mathbb{E}\left[ \mathbb{P}\left[ \Mhat_{EE^+(i)}(t) - \Mhat_{EE^+(i)}({\tau_{n}})=0 \Big \vert {\cal F}^Y_{\infty} \right] \right] \\
	& = \lim_{t\rightarrow +\infty} \mathbb{E} \left[   \exp\left(-\Lambda_{EE^+(i)}(\tau_{n},t)\right) \right]=0,
\end{align}
where to derive the last equality, we have used the assumption that $\Lambda_{EE^+(i)}(s,t)\rightarrow +\infty$ as $t\rightarrow +\infty$.

For any $E \in \Nb$ with $ \mid E\mid \, =n $, and $ i, j\in E^c $ with $i\neq j $, assertion (iii) of Proposition \ref{property_NEE} implies that $\Mhat_{EE^+(j)}$ is $\mathcal{F}^Y_{\infty}$-conditional independent of  $\Mhat_{EE^+(i)}$. Therefore,  we have $ \mathbb{P}[\tau_{n+1}(E,i)={\tau_{n+1}}(E,j)]=0$ for all $i,j\in E^c$ and $i \neq j$\footnote{Indepdent Poisson processes do not jump simultaneously.}, and thus
\begin{align}
\mathbb{P}[\mid X_{\tau_{n+1}}\mid  = n+2] &=  \sum_{E:\mid E\mid=n} \; \sum_{i,j\in E^c,i\neq j}\mathbb{P}[X_{\tau_n}=E,X_{\tau_{n+1}}=E\cup \{i,j\}]\\
&\leq  \sum_{E:\mid E\mid=n} \; \sum_{i,j\in E^c,i\neq j}\mathbb{P}[X_{\tau_n}=E,\tau_{n+1}(E,i)=\tau_{n+1}(E,j)]=0.
\end{align}
The same argument leads to   $\mathbb{P} (\mid X_{\tau_{n+1}}\mid  \geq  n+3) = 0$. Then, we conclude that $\mathbb{P}[\mid X_{\tau_{n+1}}\mid =n+1]=1$. This ends the proof.
\end{proof}

At this stage, the construction of the default process $X$ is complete.
Before we move on to show the Markov property of $X$, we present essential results in the proposition below, which are key to the proofs in the next subsection.
The  following notations are used in the sequel:
\begin{align}\label{GnLamdaE}
{\cal G}_n :=\sigma \{\tau_1, X_{\tau_1};\cdots;\tau_{n},X_{\tau_n}\}, \quad  \lambda_E(t):= -\lambda_{EE}(t), \quad \mbox{and} \quad \Lambda_E(t) := -\Lambda_{EE}(t).
\end{align}

\begin{proposition}\label{theoremcriticalforothers}
	Let Assumptions \ref{assumption_Poisson} and \ref{assumption_intensity} hold.
	The sequence $ (\tau_{n},X_{\tau_n})_{n = 0,1,\cdots,N}$, constructed using the above induction algorithm,  satisfies the following properties:
	\begin{itemize}
		\item[\rm{(i)}] $X_{\tau_n} \subseteq X_{\tau_{n+1}}$ for all $n=0,1,\cdots,N-1$.
		
		\item[\rm{(ii)}] For all $n=0,1,\cdots, N-1$ and $ t\geq 0$,
		\begin{align}
		\mathbb{P}\left[{\tau_{n+1}}-\tau_{n}>t\mid \mathcal{F}^Y_{\infty}\vee {\cal G}_n\right] &=\exp \left(-\int^{\tau_{n}+t}_{\tau_{n}}\lambda_{X_{\tau_n}}(u)\dd u \right), \\
		\text{and} \quad \mathbb{P}[{\tau_{n+1}}>t\mid \mathcal{F}^Y_{\infty}\vee{\cal
			G}_n] \cdot \mathds{1}_{\{\tau_{n}\leq
			t\}} &=\exp\left(-\int^{t}_{\tau_{n}}\lambda_{X_{\tau_n}}(u)\dd u \right) \cdot \mathds{1}_{\{\tau_{n}\leq t\}}.
		\end{align}
		
		\item[\rm{(iii)}] For all $s\geq 0$, $F\in {\mathbb{N}}$ with $\mid F\mid=n+1$, where $n=0,1,\cdots,N-1$,
		\begin{align}
		\mathbb{P}[X_{\tau_{n+1}}=F, \, \tau_{n+1}\in \dd s \mid \mathcal{F}^Y_{\infty}\vee {\cal G}_n]=\mathds{1}_{\{\tau_{n}\leq s \}}\lambda_{X_{\tau_n}F}(s)
		\exp\left(-\int^{s}_{\tau_{n}}\lambda_{X_{\tau_n}}(u)\dd u \right) \dd s,
		\end{align}
		that is,  for any measurable function $f(\cdot,\cdot)$ on ${\mathbb{N}}\times \mathbb{R}^+$,
		\begin{align}
		\mathbb{E} \left[f(X_{\tau_{n+1}},{\tau_{n+1}})\mid \mathcal{F}^Y_{\infty}\vee {\cal G}_n\right]
		=\sum_{F\in {\mathbb{N}}, \mid F \mid = n+1}\int^{+\infty}_{0}f(F,s)\mathds{1}_{\{\tau_{n}\leq s\}}\lambda_{X_{\tau_n}F}(s) \, e^{-\int^{s}_{\tau_{n}}\lambda_{X_{\tau_n}}(u)\dd u }\dd s.
		\end{align}
	\end{itemize}
\end{proposition}

\begin{proof}
(i) is obvious.

(ii) For any $E \subseteq \Nb$ with $|E|=n$, since $\{ \Mhat_{EE^{+}(i)}(t) : i \in E^c \}$ are $\mathcal{F}^Y_{\infty}$-conditional independent of ${\cal F}^{(n)}_{\infty}$, we obtain
\begin{align}
&\mathds{1}_{\{X_{\tau_n}=E\}}\mathbb{P}\left[{\tau_{n+1}}-\tau_{n}>t\mid \mathcal{F}^Y_{\infty}\vee {\cal G}_n \right]
=\mathds{1}_{\{X_{\tau_n}=E\}}\mathbb{P}\left[ \bigcap_{i \in E^c} \left\{ \tau_{n+1}(E,i)>\tau_n+t \right\}  \Big| \mathcal{F}^Y_{\infty}\vee {\cal G}_n \right]\\
= \; &\mathds{1}_{\{X_{\tau_n}=E\}}\mathbb{P}\left[ \bigcap_{i \in E^c} \left\{ \Mhat_{EE^{+}(i)} (\tau_n+t)  - \Mhat_{EE^{+}(i)}(\tau_n)=0 \right\} \Big| \mathcal{F}^Y_{\infty}\vee {\cal G}_n \right]\\
= \; &\mathds{1}_{\{X_{\tau_n}=E\}}\exp\left\{-\sum_{i\in E^c}\Lambda_{EE^+(i)}(\tau_n,\tau_n+t)\right\}
= \mathds{1}_{\{X_{\tau_n}=E\}}\exp\left\{-\Lambda_{E}(\tau_n,\tau_n+t)\right\}.
\end{align}
The second equality can be proved by following the same argument.

(iii) For $0\leq s< t$, $E\in {\mathbb N}$ with $\mid E\mid =n$, and $ i\in E^c$, we have
\begin{align}
	\quad \mathsf{P}_1:&=\mathbb{P}\Big[ \Mhat_{EE^+(i)}(s) - \Mhat_{EE^+(i)}(\tau_n)=0,  \, \Mhat_{EE^+(i)}(t) - \Mhat_{EE^+(i)}(s) > 0, \\
	&\hspace{1cm}  \Mhat_{EE^+(j)}(t) - \Mhat_{EE^+(j)}(\tau_n)=0, \, \forall j\in E^c, \, j \neq i \, \Big|  \, \mathcal{F}^Y_{\infty}\vee  {\cal G}_n\Big] \cdot \mathds{1}_{\{\tau_n\leq s, \, X_{\tau_n}=E\}}\\
	&\leq  \mathbb{P}[X_{\tau_{n+1}}=E^+(i),s<{\tau_{n+1}}\leq t\mid \mathcal{F}^Y_{\infty}\vee  {\cal G}_n] \mathds{1}_{\{\tau_n\leq s, \, X_{\tau_n}=E\}}\\
	&\leq  \mathbb{P} \Big[ {  \Mhat_{EE^+(j)}(s) - \Mhat_{EE^+(j)}(\tau_n)=0, \, \forall j\in E^c, \, j \neq i}, \, \\
	&\hspace{1cm} \Mhat_{EE^+(i)}(t) - \Mhat_{EE^+(i)}(s) > 0 \, \Big| \, \mathcal{F}^Y_{\infty}\vee  {\cal G}_n \Big]  \cdot \mathds{1}_{\{\tau_n\leq s, \, X_{\tau_n}=E\}} := \mathsf{P}_2.
\end{align}
{  Since $F \in \Nb$ and $|F| = n+1$, we have $F = E^+(i) = E \cup \{i\}$ for some $E$ and $i \in E^c$.}
Denote
\begin{align}
	\mathsf{p}_1(s,t; F) &= e^{-\Lambda_{EE^+(i)}(\tau_n,s)} \left(1-e^{-\Lambda_{EE^+(i)}(s,t)} \right)
	\prod_{j\in E^c, \, j\neq i}\exp\{-\Lambda_{EE^+(j)}(\tau_n,t)\},\\
	\mathsf{p}_2(s,t; F) &= \left(1-e^{-\Lambda_{EE^+(i)}(s,t)} \right) \prod_{j\in E^c, \, {  j\neq i}} \exp\{-\Lambda_{EE^+(j)}(\tau_n,s)\}.
\end{align}
It is easy to see $\mathsf{P}_1  = \mathds{1}_{\{\tau_n\leq s, \, X_{\tau_n}=E\}} \cdot \mathsf{p}_1(s,t;F) $ and $\mathsf{P}_2 = \mathds{1}_{\{\tau_n\leq s, \, X_{\tau_n}=E\}} \cdot \mathsf{p}_2(s,t;F).$
Hence, on the set $\{\tau_n\leq s, \, X_{\tau_n}=E\}$, we obtain
$\mathsf{p}_1(s,t;F) \leq  \mathbb{P}[X_{\tau_{n+1}}=F, \, s< {\tau_{n+1}}\leq t\mid \mathcal{F}^Y_{\infty}\vee  {\cal G}_n] \leq  \mathsf{p}_2(s,t;F)$.
By the existence   of regular conditional probability, there exists a random   measure $p^{(n)}$, where $p^{(n)}(\omega, A): \Omega\times{\mathbb{N}}\times {\cal
	B}\rightarrow [0,1],$ such that $p^{(n)}(\omega,F\times (s,t])$ is equal to $\mathbb{P}[X_{\tau_{n+1}}=F, \, s<{\tau_{n+1}}\leq t \mid \mathcal{F}^Y_{\infty}\vee  {\cal G}_n](\omega)$.
Since
\begin{align}
\lim_{t\downarrow s}\frac{\mathsf{p}_1(s,t;F)}{t-s}
	=\lim_{t\downarrow s}\frac{\mathsf{p}_2(s,t;F)}{t-s} =\lambda_{X_{\tau_n}F}(s)\exp\left\{-\int^s_{\tau_n}\lambda_{X_{\tau_n}}(u)\dd u\right\},
\end{align}
we obtain,  for all $s> \tau_n(\omega)$, that
$$
\frac{\dd p^{(n)}(\omega,F\times \dd s)}{\dd s}
=\lambda_{X_{\tau_n}F}(s)\exp\left\{-\int^{s}_{\tau_n}\lambda_{X_{\tau_n}}(u)\dd u \right\}.
$$
Therefore, for any measurable function $f(\cdot,\cdot)$  on ${\mathbb{N}}\times \mathbb{R}^+$,
\begin{align}
	&\mathbb{E}[f(X_{\tau_{n+1}},{\tau_{n+1}})\mid \mathcal{F}^Y_{\infty}\vee{\cal G}_n]
	=\sum_{F\in {\mathbb{N}}, \ |F|=n+1}\int^{\infty}_{\tau_n}f(F,s)p^{(n)}(\omega,F\times \dd s)\\
	= \; &\sum_{F\in {\mathbb{N}}, \ |F|=n+1}\int^{\infty}_{\tau_n}f(F,s)\lambda_{X_{\tau_n}F}(s)\exp\left\{-\int^{s}_{\tau_n}\lambda_{X_{\tau_n}}(u)\dd u \right\} \dd s\\
	= \; & \sum_{F\in {\mathbb{N}}, \ |F|=n+1}\int^{\infty}_{0}f(F,s)\mathds{1}_{\{\tau_n\leq s\}}\lambda_{X_{\tau_n}F}(s)\exp\left\{-\int^{s}_{\tau_n}\lambda_{X_{\tau_n}}(u)\dd u\right\} \dd s,
\end{align}
which completes the proof.
\end{proof}

\subsection{Conditional Markov Property and Transition Probability of $X$}
\label{appen_markov}
In this subsection, our main goal is to show the conditional Markov property of the default process $X$, and characterize its transition probability. The related conclusions are presented in Proposition \ref{maintheoremmarkovG}.

\begin{proposition}\label{maintheoremmarkovG}
	Let Assumptions \ref{assumption_Poisson} and \ref{assumption_intensity} hold.
	The default process $X$, constructed as in Appendix \ref{appen_construction}, satisfies the following properties:
	\begin{itemize}
	\item[\rm{(i)}]  For any $0=t_0\leq t_1<\cdots<t_n,$ and arbitrary sets $\emptyset = F_0 \subseteq F_1\subseteq\cdots \subseteq
	F_n \in {\mathbb{N}}$, we have
	\begin{align}\label{dynamicofXinappendix}
	\mathbb{P} \left[\bigcup\limits_{i=1}^n \{X_{t_i}=F_i\}  \Big| \mathcal{F}^Y_{\infty} \right]
		=\mathbb{P}\left[\bigcup\limits_{i=1}^n \{X_{t_i}=F_i\} \Big| \mathcal{F}^Y_{t_n} \right] = \prod^{n-1}_{i=0}G(t_i,t_{i+1};F_i,F_{i+1}),
	\end{align}
	where $G$ is defined in \eqref{definitionFunctionG}.
	\item[\rm{(ii)}]  (Markov Property).  For any $ F\in {\mathbb{N}}$ and $0\leq s < t$,  we have
	\begin{align}
	\mathbb{P} \left[X_{t}=F\mid {\cal F}^X_s \vee \mathcal{F}^Y_{t} \right] =\mathbb{P} \left[X_{t}=F\mid \sigma(X_s) \vee \mathcal{F}^Y_{t} \right].
	\end{align}
\end{itemize}
\end{proposition}

\begin{proof}
(i) Recall ${\cal G}_n$ defined in \eqref{GnLamdaE}. We first show, for any $(E,F)\in {\mathbb{N}}^2_+$, $A\in {\cal G}_{\mid
	E\mid},$ and $0 \le s< t$,
\begin{align}\label{eqXSXT1}
\mathbb{P}\left[A \{X_s=E,X_t=F \}\mid \mathcal{F}^Y_{\infty}\right] = G(s,t; E,F) \cdot \mathbb{P}\left[A \{ X_s=E \}\mid \mathcal{F}^Y_{\infty} \right],
\end{align}
where $G$ is defined in \eqref{definitionFunctionG}.  Suppose   $\mid E\mid  =m, \mid F/E\mid=n$, where $0\leq m, n\leq N$. We   prove \eqref{eqXSXT1} by  induction.

{\bf Step 1:} If $n=0$, i.e., $E=F$. By assertion (ii) of Proposition \ref{theoremcriticalforothers}, we have
\begin{align}
	& \mathbb{P} \left[A\{X_s=E, \, X_t=F\} \mid \mathcal{F}^Y_{\infty} \right]\\
	=\; &\mathbb{E} \left[ \mathds{1}_{A\{\tau_m\leq s, \, X_{\tau_m}=E\}} \cdot \mathbb{P} \left(\tau_{m+1}>t\mid \mathcal{F}^Y_{\infty}\vee {\cal G}_m \right) \big| \mathcal{F}^Y_{\infty} \right]  \\
	=\; &\exp\left(-\int^t_{s}\lambda_E(u)\dd u \right) \cdot \mathbb{E} \left[\mathds{1}_{A\{\tau_m\leq s, \, X_{\tau_m}=E\}} \exp\left(-\int^s_{\tau_m}\lambda_E(u)\dd u \right) \Big| \mathcal{F}^Y_{\infty}\right]\\
	=\; &H_0(s,t;E) \cdot \mathbb{E} \Big[\mathds{1}_{A\{\tau_m\leq s, \, X_{\tau_m}=E\}} \, \mathbb{P}[\tau_{m+1}>s\mid \mathcal{F}^Y_{\infty}\vee {\cal G}_m] \Big| \mathcal{F}^Y_{\infty}\Big]\\
	=\; &H_0(s,t;E) \cdot \mathbb{P}\Big[A\{\tau_m\leq s<\tau_{m+1}, \, X_{\tau_m}=E \} \Big| \mathcal{F}^Y_{\infty}\Big]\\
	=\; &G(s,t;E,F) \cdot \mathbb{P}[A\{ X_s=E \}\mid \mathcal{F}^Y_{\infty}].
\end{align}

{\bf Step 2:} Suppose that (\ref{eqXSXT1}) holds for all  pairs $(E,F) \in \Nb^2_{++}$ with $| F/E| = k$, where $0\leq k<N$.
Now consider a pair $(E,F)\in \Nb^2_{++}$ with $|F/E| =k+1$. By  (ii) and (iii) of Proposition \ref{theoremcriticalforothers}, we deduce
\begin{align}
&\mathbb{P}[A \{X_s=E,X_t=F \}\mid \mathcal{F}^Y_{\infty}] \\
=\; &\mathbb{E}\left[\mathds{1}_{A \{X_s=E, \, \tau_{m+k+1}\leq t, \, X_{m+k+1}=F \}} \, \mathbb{P} \big[\tau_{m+k+2}>t \big| \mathcal{F}^Y_{\infty}\vee{\cal G}_{m+k+1} \big] \Big| \mathcal{F}^Y_{\infty}\right] \\
=\; &\mathbb{E}\left[\mathds{1}_{A \{ X_s=E\}} \mathds{1}_{\{ X_{\tau_{m+k+1}}=F, \, \tau_{m+k+1}\leq t\}}
\exp\left\{-\int^t_{\tau_{m+k+1}}\lambda_F(u)\dd u\right\} \Bigg| \mathcal{F}^Y_{\infty}\right]\\
=\; &\sum_{i\in F/E}\mathbb{E}\Big[\mathds{1}_{A\{X_s=E,\, X_{\tau_{m+k}} = F^-(i)\}}\int^t_s\mathds{1}_{\{ \tau_{m+k}\leq h \}}\lambda_{F^-(i)F}(h) \\
 & \; \times \exp\left\{-\int^h_{\tau_{m+k}}\lambda_{F^-(i)}(u)\dd u
-\int^t_h\lambda_{F}(u)\dd u\right\}\dd h \Big| \mathcal{F}^Y_{\infty}\Big] \\
=\; & \sum_{i\in F/E}\int^t_s\lambda_{F^{-}(i)F}(h)\exp\{-\int^t_h\lambda_{F}(u)\dd u\} \\
& \; \times \mathbb{E}\left[\mathds{1}_{A\{X_s=E, \, \tau_{m+k}\leq h, \, X_{\tau^{m+k}}=F^{-}(i)\}}\exp\left\{-\int^h_{\tau^{m+k}}\lambda_{F^{-}(i)}(u)\dd u\right\} \Big| \mathcal{F}^Y_{\infty}\right] \dd h \\
=\; & \sum_{i\in F/E}\int^t_s\lambda_{F^{-}(i)F}(h)\exp\left\{-\int^t_h\lambda_{F}(u)\dd u\right\} \\
& \; \times \mathbb{E}\Big[\mathds{1}_{A\{X_s=E, \, \tau_{m+k}\leq h, \, X_{\tau_{m+k}}=F^{-}(i)\}}\mathbb{P}\left[\tau_{m+k+1}>h \Big| {\cal G}_{m+k}\vee\mathcal{F}^Y_{\infty}\right] \Big] \dd h \\
= \; & \sum_{i\in F/E}\int^t_s\lambda_{F^{-}(i)F}(h)\exp\left\{-\int^t_h\lambda_{F}(u)\dd u \right\} \, \mathbb{P}\left[A\{X_s=E, \, X_h=F^{-}(i)\}\mid \mathcal{F}^Y_{\infty}\right] \dd h  \label{eqinportmantinducetion1} \\
= \; &\sum_{i\in F/E}\int^t_s\lambda_{F^{-}(i)F}(h)\exp\left\{-\int^t_h\lambda_{F}(u)\dd u \right\} \sum_{\pi\in \Pi(F^{-}(i)/E)}H_k(s,h; F^{\pi}_0,\cdots,F^{\pi}_k)\mathbb{P}\left[A \{X_s=E\}\mid \mathcal{F}^Y_{\infty}\right] \dd h\\
= \; &\mathbb{P}[A\{X_s=E\}\mid \mathcal{F}^Y_{\infty}]  \sum_{i\in F/E}\sum_{\pi\in \Pi(F^{-}(i)/E)}\int^t_s\lambda_{F^{-}(i)F}(h)\exp\left\{-\int^t_h\lambda_{F}(u)\dd u\right\}H_k(s,h; F^{\pi}_0,\cdots,F^{\pi}_k) \dd h\\
= \; &\mathbb{P}[A\{X_s=E\}\mid \mathcal{F}^Y_{\infty}]  \sum_{\pi\in \Pi(F/E)}H_{k+1}(s,h; F^{\pi}_0,\cdots,F^{\pi}_{k+1})\\
= \; & G(s,t;E,F) \, \mathbb{P}[A\{X_s=E\}\mid\mathcal{F}^Y_{\infty}],
\end{align}
which shows equality \eqref{eqXSXT1} holds true for $\mid F/E\mid =k+1$.

Now by taking $A$ in \eqref{eqXSXT1} as $A=\{ X_{t_i}=F_i,i=1,\cdots,n-2\}\in {\cal G}_{\mid F_{n-1} \mid}$, we  get
\begin{align}
	\mathbb{P}\left[\bigcup\limits_{i=1}^n \{X_{t_i}=F_i\}  \Bigg| \mathcal{F}^Y_{\infty} \right]   &= \mathbb{P}[A \{X_{t_{n-1}}=F_{n-1},X_{t_n}=F_n\}\mid\mathcal{F}^Y_{\infty}]\\
	&=G(t_{n-1},t_n;F_{n-1},F_n) \, \mathbb{P} \left[A \{X_{t_{n-1}}=F_{n-1}\} \big| \mathcal{F}^Y_{\infty}\right]\\
	&=G(t_{n-1},t_n;F_{n-1},F_n) \, \mathbb{P} \left[\bigcup_{i=1}^{n-1} X_{t_i}=F_i  \Bigg| \mathcal{F}^Y_{\infty} \right] \\
	&= \prod^{n-1}_{i=0}G(t_i,t_{i+1};F_i,F_{i+1}).
\end{align}
The second equality of \eqref{dynamicofXinappendix} follows from the fact that $\prod^{n-1}_{i=0}G(t_i,t_{i+1};F_i,F_{i+1})$ is ${\cal F}^Y_{t_n}$-measurable. This completes the proof of assertion (i).

(ii) Notice that, for any $F \in \Nb$ and $0 \le s <t$, using \eqref{dynamicofXinappendix} in (i), we derive
\begin{align}
	\mathbb{P} \left[X_{t}=F \big| \sigma(X_s) \vee \mathcal{F}^Y_{t} \right] &=\sum_{E\subseteq F}\mathds{1}_{ \{  X_s=E \} }\dfrac{\mathbb{P}[X_s=E, \, X_{t}=F\mid \mathcal{F}^Y_{t}]}{\mathbb{P}[X_s=E \mid \mathcal{F}^Y_{t}]}\\
	&= \sum_{E\subseteq F}\mathds{1}_{\{X_s=E\}} \cdot G(s,t;E,F).
\end{align}
Hence, it is enough to prove, for all $(E,F)\in {\mathbb{N}}^2_+$, that
\begin{align}
	\mathds{1}_{\{X_s=E \}}\mathbb{P}[X_{t}=F \mid \F_s^X \vee \mathcal{F}^Y_{t}]
	=\mathds{1}_{\{ X_s=E\}} \cdot G(s,t;E,F),
\end{align}
or, equivalently, $\forall  \ n\geq 1$, $\forall\ s \le t_1<t_2<\cdots<t_n\leq t$,
 $E_1\subseteq \cdots\subseteq E_n\subseteq E$, and $B\in \mathcal{F}^Y_{t},$
\begin{align}
	\mathbb{E}\Big[\mathds{1}_{B\{X_s=E, \, X_{t_i}=E_i, \, i=1,\cdots,n\}}G(s,t;E,F)\Big] = \mathbb{P}\left[B\{X_s=E,X_{t_i}=E_i,i=1,\cdots,n,X_{t}=F\}\right],
\end{align}
which is obvious by (i). The proof is then complete.
\end{proof}

\subsection{Martingale Property of $X$}
\label{appen_martingale}
In this subsection, we complete the last part of the proof to Theorem \ref{theorem_existence} by showing that the family $\Lamb:=(\Lambda_{EF}(t))_{t\geq 0}$ as specified by Assumption \ref{assumption_intensity} is the default intensity of the default process $X$.
The key results are summarized in the proposition below.

\begin{proposition}\label{propIndensity}
	For any $ F\in {\mathbb{N}}$, the process $X_F = (X_F(t))_{t \ge 0}$, defined by
	\begin{align}
	X_{F}(t) := \mathds{1}_{F}(X_t)-\int^t_0\lambda_{X_uF}(u)\dd u,
	\end{align}
	{ is an $\check{\Fb}$-martingale, where $\check{\Fb} = (\check{\F}_t)_{t \ge 0}:= (\F_t^X \vee \F_t^Y )_{t \ge 0}$.}
\end{proposition}

\begin{proof}
	It is enough to prove, for all $0\leq s\leq t$,  $E\in {\mathbb{N}}$, with $E\subseteq F$, that  $\mathds{1}_{\{X_s=E\}}\mathbb{E}[X_{F}(t)-X_{F}(s)\mid \check{\F}_s]=0,$ i.e.,
	\begin{align}
		\mathds{1}_{\{X_s=E\}}\mathbb{E}[\mathds{1}_{\{X_t=F\}}-\mathds{1}_{  \{X_s=F\}}\mid \check{\F}_s] = \mathds{1}_{\{X_s=E\}}\mathbb{E}\left[\int^t_s\lambda_{X_uF}(u)\dd u \Big| \check{\F}_s \right].
	\end{align}
	By the monotone class theorem, for any $s_1<s_2<\cdots <s_n<s$ and $ E_1\subseteq \cdots\subseteq E_n\subseteq E$, without loss of generality, we take  arbitrary $ A\in {\cal F }^Y_s$, $B=\{X_{s_i}=E_i, \, i=1,\cdots,n\}$, and show that
	\begin{equation}\label{eqmar2}
	\mathbb{E} \left[\mathds{1}_{AB\{ X_s=E  \}}\int^t_s\lambda_{X_uF}(u)\dd u \right] =\mathbb{E}\left[\mathds{1}_{AB} \left(\mathds{1}_{\{ X_s=E, \, X_t=F \}}-\mathds{1}_{  \{X_s=E, \, X_s=F \}} \right)\right].
	\end{equation}
	In the following, we will prove \eqref{eqmar2} for the cases $E=F$ and $E\subset F$. Recall   functionals $G$ and $H$ defined in \eqref{definitionFunctionG} and \eqref{defFuntionH}.
	
	\noindent
	{\bf Case 1:} $E\subset F$. Without loss of generality, we assume   $|F| = |E|+m $ and $m\geq
	1.$ Then, using assertions (i) and (ii) in Proposition \ref{maintheoremmarkovG},  we have
	\begin{eqnarray}
	&&\mathbb{E}\left[\mathds{1}_{AB\{X_s=E\}}\int^t_s\lambda_{X_uF}(u)\dd u \right] \\
	&=&\sum_{k\in F/E}\, \int^t_s\mathbb{E}\left[\mathds{1}_A \mathbb{E} \left[\mathds{1}_{B\{ X_s=E,X_u=F^{-}(k)\}} \Big| \mathcal{F}^Y_{\infty}\right] \lambda_{F^{-}(k)F}(u)\right]\dd u \\
	&&+\int^t_s \mathbb{E} \left[\mathds{1}_A \mathbb{E} \left[\mathds{1}_{B\{X_s=E,X_u=F\}} \Big| \mathcal{F}^Y_{\infty}\right] \lambda_{FF}(u)\right]\dd u \\
	&=&\sum_{k\in F/E}\mathbb{E} \left[\mathds{1}_AG(0,s_1;\emptyset,E_1)\cdots G((s_n,s;E_n,E)\int^t_s\lambda_{F^{-}(k)F}(u))G(s,u;E,F^-(k))\dd u \right] \\
	&&+\mathbb{E} \left[\mathds{1}_AG(0,s_1;\emptyset,E_1)\cdots G((s_n,s;E_n,E)\int^t_s\lambda_{FF}(u))G(s,u;E,F)\dd u\right].
	\end{eqnarray}
	Take any $\bm{\pi}=(\pi_1,\cdots,\pi_m)\in \Pi(F/E)$, and recall the definitions of $(F^{\bm{\pi}}_k)_{k=0,1,\cdots,m}$ in \eqref{eqFPI} and $\lambda_{FF}(u)=-\lambda_{F}(u)$, we obtain
	\begin{eqnarray}
		&&\int^t_s H_m(s,u;F^{\bm{\pi}}_0,\cdots,F^{\bm{\pi}}_m)\lambda_{FF}(u)\dd u\\
		&=&\int^t_s\lambda_{FF}(u)\int^u_s\lambda_{F^{\bm{\pi}}_{m-1}F^{\bm{\pi}}_m}(v)
		\exp\left\{-\int^u_v\lambda_{F^{\bm{\pi}}_m}(l)dl \right\} H_{m-1}(s,v;F^{\bm{\pi}}_0,\cdots,F^{\bm{\pi}}_{m-1}) \dd v \dd u\\
		&=&\int^t_s\lambda_{F^{\bm{\pi}}_{m-1}F^{\bm{\pi}}_m}(v)H_{m-1}(s,v;F^{\bm{\pi}}_0,\cdots,F^{\bm{\pi}}_{m-1})\int^t_v\exp \left\{-\int^u_v\lambda_{F}(l) \dd l \right\}\lambda_{FF}(u)) \dd u \dd v\\
		&=&\int^t_s\lambda_{F^{\bm{\pi}}_{m-1}F^{\bm{\pi}}_m}(v)H_{m-1}(s,v;F^{\bm{\pi}}_0,\cdots,F^{\bm{\pi}}_{m-1})\left(\exp\{-\int^t_v\lambda_{F}(l) \dd l\}-1\right)dv\\
		&=&\int^t_s\lambda_{F^{\bm{\pi}}_{m-1}F^{\bm{\pi}}_m}(v)\exp \left\{-\int^t_v\lambda_{F}(l) \dd l\right\} H_{m-1}(s,v;F^{\bm{\pi}}_0,\cdots,F^{\bm{\pi}}_{m-1}) \dd v\\
		&&-\int^t_s\lambda_{F^{\bm{\pi}}_{m-1}F^{\bm{\pi}}_m}(v)H_{m-1}(s,v;F^{\bm{\pi}}_0,\cdots,F^{\bm{\pi}}_{m-1}) \dd v\\
		&=&H_m(s,t;F^{\bm{\pi}}_0,\cdots,F^{\bm{\pi}}_m)-\int^t_s\lambda_{F^{\bm{\pi}}_{m-1}F}(v)H_{m-1}(s,v;F^{\bm{\pi}}_0,\cdots,F^{\bm{\pi}}_{m-1}) \dd v.
	\end{eqnarray}
	Thus,
	\begin{eqnarray}
		&&\int^t_s\lambda_{FF}(u)G(s,u;E,F)\dd u = \sum_{\pi\in \Pi(F/E)}\int^t_s H_m(s,u;F^{\bm{\pi}}_0,\cdots,F^{\bm{\pi}}_m)\lambda_{FF}(u)\dd u\\
		&=&\sum_{\pi\in \Pi(F/E)} H_{m}(s,t;F^{\bm{\pi}}_{0}\cdots,F^{\bm{\pi}}_{m}) -\sum_{\pi\in \Pi( F/E)}\int^t_s\lambda_{F^{\bm{\pi}}_{m-1}F}(v)H_{m-1}(s,v;F^{\bm{\pi}}_0,\cdots,F^{\bm{\pi}}_{m-1}) \dd v\\
		&=&G(s,t;E,F)-\sum_{k\in F/ E}\sum_{\pi\in \Pi(F^-(k)/E)}\int^t_s\lambda_{F^-(k)F}(v)H_{m-1}(s,v;F^{\bm{\pi}}_0,\cdots,F^-(k)) \dd v\\
		&=&G(s,t;E,F)-\sum_{k\in F/ E}\int^t_s\lambda_{F^{-}(k)F}(v)G(s,v;E,F^{-}(k)) \dd v.
	\end{eqnarray}
	Finally, we are able to show that
	\begin{eqnarray}
		&&\mathbb{E}\left[\mathds{1}_{AB\{ X_s=E\}}\int^t_s\lambda_{X_uF}(u)\dd u \right]\\
		&=&\mathbb{E}\left[\mathds{1}_AG(0,s_1;\phi,E_1)G(s_1,s_2;E_1,E_2)\cdots G(s_n,s;E_n,E)G(s,t;E,F)\right]\\
		&=&\mathbb{E}\left[\mathds{1}_A \, \mathbb{P}[X_{s_i}=E_i, \, i=1,\cdots,n,X_s=E,X_t=F\mid \mathcal{F}^Y_{\infty}] \right]\\
		&=&\mathbb{P}[AB\{X_s=E, \, X_t=F\}],
	\end{eqnarray}
	which completes the proof for the case of $E\subset F$.
	
	\noindent
	{\bf Case 2:} $E=F$. Since  $G(s,t;E,E)=H_0(s,t;E)=\exp\{-\int^t_s\lambda_E(u) \dd u\}$, we derive
	\begin{eqnarray*}
		&&\mathbb{E}\left[\mathds{1}_{AB\{ X_s=E\}}\int^t_s\lambda_{X_uF}(u)\dd u\right] =  \int^t_s \mathbb{E}\left[\mathds{1}_{AB\{ X_s=E,X_u=E\}}\lambda_{EE}(u)\right]\dd u\\
		&=&\mathbb{E}\left[\mathds{1}_{A}G(0,s_1;\emptyset,E_1)\cdots G(s_n,s;E_n,E)\int^t_s\lambda_{EE}(u)G(s,u;E,E)\dd u\right]\\
		&=&\mathbb{E}\left[\mathds{1}_{A}G(0,s_1;\emptyset,E_1)\cdots G(s_n,s;E_n,E)\left(\exp\left\{-\int^t_s\lambda_E(l)dl\right\}-1\right)\right]\\
		&=&\mathbb{E}\left[\mathds{1}_{AB\{X_s=E,X_t=F\}}\right] - \mathbb{E}\left[\mathds{1}_{AB\{X_s=E\}}\right] \\
		&=&\mathbb{E}\left[\mathds{1}_{AB}(\mathds{1}_{\{X_s=E,X_t=F\}}-\mathds{1}_{\{X_s=E\}})\right].
	\end{eqnarray*}
	This proves the case of $E=F$, and thus completes the proof of the proposition.
\end{proof}

\section{Technical Proofs}
\label{sec_proofs}

\subsection{Proof of Theorem \ref{maintheoremforcaluclation}}
\label{appen_maintheoremforcaluclation}
\begin{proof}
From assertion (i) of  Proposition \ref{maintheoremmarkovG}, we deduce, for any $F \in \Nb$ and $0 \le s < t$, that
\begin{align}
\mathbb{P}[X_t=F\mid \F^X_s \vee \F^Y_t] &=\sum_{E\subseteq F}\mathds{1}_{\{X_s=E\}}\frac{\mathbb{P}[X_s=E,X_t=F\mid \mathcal{F}^Y_t]}{\mathbb{P}[X_s=E\mid\mathcal{F}^Y_t]}\\
&= \sum_{E\subseteq F}\mathds{1}_{\{X_s=E\}} \cdot G(s,t;E,F),
\end{align}
which proves (\ref{eqmain1}).

To show \eqref{eqmain2}, notice that, for any bounded $\mathcal{F}^Y_{t}$-measurable $\xi$,
 \begin{eqnarray*}
 &&\mathbb{E}\Big[\mathds{1}_{\{X_{t}=F\}}\xi \Big| \F^X_s \vee \F^Y_s \Big]=\mathbb{E}\Big[ \mathbb{E} \Big [\mathds{1}_{\{X_{t}=F\}}\xi \Big| \F^X_s \vee \F^Y_t \Big] \Big| \F^X_s \vee \F^Y_s \Big ]\\
 &=&\sum_{E\subseteq F}\mathds{1}_{\{X_s=E\}}\mathbb{E}\Big[\xi G(s,t;E,F)\Big\vert  \F^X_s \vee \F^Y_s \Big] = \sum_{E\subseteq F}\mathds{1}_{\{X_s=E\}}\frac{\mathbb{E}\Big[\xi  G(s,t;E,F)\mathds{1}_{\{X_s=E\}}\Big\vert\mathcal{F}^Y_s\Big]}{\mathbb{P}[X_s=E |  \mathcal{F}^Y_s]}\\
   &=&\sum_{E\subseteq F}\mathds{1}_{\{X_s=E\}}\frac{\mathbb{E}\Big[\xi G(s,t;E,F)G(0,s;\emptyset,E)\Big\vert\mathcal{F}^Y_s\Big]}{G(0,s;\emptyset,E)}=\sum_{E\subseteq F}\mathds{1}_{\{X_s=E\}}\mathbb{E}\Big[\xi  G(s,t;E,F)\Big\vert\mathcal{F}^Y_s\Big].
\end{eqnarray*}
This ends the proof of Theorem \ref{maintheoremforcaluclation}.
\end{proof}

The proof of Corollay \ref{theoremparticular} relies on the following lemma.
\begin{lemma}
\label{lemmaHalphafunction}
Let {\blue $z(\cdot)$} be a nonnegative  function defined on $\mathbb{R}^+$ and {\blue $Z(\cdot, \cdot)$} be a nonnegative function on $\mathbb{R}^+ \times \mathbb{R}^+$ given by $Z(s,t) :=  \int^t_s z(h)\dd h$.
Let $l_0,l_1,\cdots,l_n$ be $n+1$ real numbers, where $n$ is a positive integer.
For any $ 0\leq s<t$, define the sequence $({\cal H}_m)_{m=0,1,\cdots,n}$ by
\begin{align}
{\cal H}_m(s,t;l_0,\cdots,l_m)= \int^t_se^{-l_m \, Z(u,t)} \cdot {\cal H}_{m-1}(s,u;l_0,\cdots,l_{m-1})\dd Z(s,u), \quad m=1,2,\cdots,n,
\end{align}
and ${\cal H}_0(s,t;l_0)=e^{-l_0 \, Z(s,t)}$.

Then ${\cal H}_m$ can be reduced to
\begin{equation}\label{eqHfunction}
{\cal H}_m(s,t;l_0,\cdots,l_m)=\sum^m_{i=0}\alpha^{(m)}_i(l_0,\cdots,l_m)e^{-l_i \, Z(s,t)},
\end{equation}
where $(\alpha^{(m)}_i)$s are defined in \eqref{eqn_alpha}.
\end{lemma}

\begin{proof}
We prove this lemma by induction.
The base case of $k=0$ is trivial.

Next suppose \eqref{eqHfunction} holds true for $k=0,1,\cdots,m$, where $m<n$.
Regarding ${\cal H}_{m+1}$, we have
\begin{align}
&\;{\cal H}_{m+1}(s,t;l_0,\cdots,l_{m+1})=\int^t_se^{-l_{m+1} \, Z(u,t)}{\cal H}_m(s,u;l_0,\cdots,l_m)\dd Z(s,u)\\
=&\;\sum^m_{i=0}\alpha^{(m)}_i(l_0,\cdots,l_m)\int^t_se^{-l_{m+1} \, Z(u,t)}e^{-l_i \, Z(s,u)}\dd Z(s,u)\\
=&\;\sum^m_{i=0}\alpha^{(m)}_i(l_0,\cdots,l_m)e^{-l_{m+1} \, Z(s,t)}\int^t_se^{-(l_i-l_{m+1})Z(s,u)}\dd Z(s,u)\\
=&\;\sum^m_{i=0}\alpha^{(m)}_i(l_0,\cdots,l_m)e^{-l_{m+1}Z(s,t)} \frac{1}{l_{m+1} - l_i} \left( e^{-(l_i-l_{m+1})Z(s,t)}-1 \right) \\
=&\;\sum^m_{i=0}  \frac{1}{l_{m+1} - l_i} \alpha^{(m)}_i(l_0,\cdots,l_m)e^{-l_{i}Z(s,t)}-\sum^m_{i=0} \frac{1}{l_{m+1} - l_i} \alpha^{(m)}_i(l_0,\cdots,l_m)e^{-l_{m+1}Z(s,t)}\\
=&\;\sum^m_{i=0}\alpha^{(m+1)}_i(l_0,\cdots,l_m,l_{m+1})e^{-l_{i}Z(s,t)}- \sum^m_{i=0}\alpha^{(m+1)}_i(l_0,\cdots,l_{m+1}) e^{-l_{m+1}Z(s,t)}\\
=&\;\sum^{m+1}_{i=0}\alpha^{(m+1)}_i(l_0,\cdots,l_m,l_{m+1})e^{-l_{i}Z(s,t)},
\end{align}
where in the last equality, we used the equation $\alpha^{(m+1)}_{m+1}:=-\sum^m_{i=0}\alpha^{(m+1)}_i(l_0,\cdots,l_{m+1})$ in \eqref{eqn_alpha}. The proof is then complete.
\end{proof}

\subsection{Proof of Corollary  \ref{theoremparticular}}
\label{appen_theoremparticular}

\begin{proof}
It is easy to see, for all $(E,F)\in {\mathbb{N}}_{++}^2 $ with $\mid F/E\mid=n$ and $ \bm{\pi} \in \Pi(F/E)$, that
\begin{align}
\lambda_{F^{\bm{\pi}}_kF^{\bm{\pi}}_{k+1}}(t)=\Lc_{F^{\bm{\pi}}_k}(\pi_{k+1})\Phi(t,Y_t) \ \ \mbox{and} \ \  \lambda_{F^{\bm{\pi}}_k}(t)= -\lambda_{F^{\bm{\pi}}_k F^{\bm{\pi}}_k}(t)=\overline{\Lc}_{F^{\bm{\pi}}_k}\Phi(t, Y_t).
\end{align}
From \eqref{eqHfunction} of Lemma \ref{lemmaHalphafunction}, by setting $z(t)=\Phi(t,Y_t)$ for all $t \ge 0$ and $l_i = \overline{\Lc}_{ {F}^{\bm \pi}_i}$, we obtain
\begin{eqnarray*}
H_n(s,t;F^{\bm{\pi}}_0,\cdots, F^{\bm{\pi}}_n)
&=&\prod^{n-1}_{k=0}\Lc_{F^{\bm{\pi}}_k}(\pi_{k+1}){\cal H}_n(s,t;  \overline{\Lc}_{F^{\bm{\pi}}_0},\cdots, \overline{\Lc}_{F^{\bm{\pi}}_n})\nonumber\\
&=&\widehat{\Lc}^{\pi}(n)\sum^{n}_{i=0}\alpha^{(n)}_i( \overline{\Lc}_{F^{\bm{\pi}}_0},\cdots, \overline{\Lc}_{F^{\bm{\pi}}_n})\exp\left({- \overline{\Lc}_{F^{\bm{\pi}}_i} {\mathcal I}(s,t)}\right)\nonumber\\
&=&\widehat{\Lc}^{\pi}(n)\sum^{n}_{i=0}\alpha^{(n)}_i(\pi)\exp\left({- \overline{\Lc}_{F^{\bm{\pi}}_i}{\mathcal I}(s,t)}\right),
\end{eqnarray*}
where $\mathcal{I}(s,t) = \int_s^t \Phi(u, Y_u) \dd u$.
The above equality, together with the results from Theorem \ref{maintheoremforcaluclation}, completes the proof to Corollary  \ref{theoremparticular}.
\end{proof}

\subsection{Proof of Proposition \ref{propExampleHomo}}
\label{appen_propExampleHomo}

\begin{proof}
For any $F \in \Nb$ with $|F|=n$, and $ \bm{\pi}=(\pi_1,\cdots,\pi_n)\in\Pi(F/\emptyset)$, we have
\begin{align}
\widehat{\Lc}^{\bm \pi}(n)=(n-1)!\rho^{n-1}e^{-\delta n(n-1)/2} { \beta_{\pi_1 }} \quad \text{and} \quad
\overline{\Lc}_{F^{\bm{\pi}}_k}=a_k, \quad k=0,1,\cdots,N.
\end{align}
In addition, recall $\mathcal{A}(n) = \{ F \in \Nb: |F| =n\} $, we have
\begin{align}
\sum\limits_{F \in \mathcal{n}} \; \sum\limits_{{\bm \pi} \in \Pi(F/\emptyset)}\beta_{\pi_1} = \sum\limits_{i=1}^{N}\beta_i (n-1)! \, C^{n-1}_{N-1} = a_0 \dfrac{(N-1)!}{(N-n)!},
\end{align}
and, for any doubled indexed sequence   $(w_j^n)_{j,n}$,
\begin{align}
\sum_{n= \parallel V_{i-1} \parallel +1}^{N} \; \sum_{j=0}^{n} \; w_j^n =
\sum_{j=0}^{N} \; \sum_{n=\max(j, \parallel V_{i-1} \parallel+1)}^{N}w^n_j,
\end{align}
where $V_{i-1} = \frac{N}{1-R} p_{i-1}$ and $\parallel V_{i-1} \parallel$ is the integer part of $V_{i-1}$. Recall $I^{(i)}(\cdot)$ is defined by \eqref{eqn_I} in Proposition \ref{proposition_spread}. Using the above results, we derive
\begin{eqnarray*}
   \mathbb{E}[L^{(i)}(X_{t_{k}})]  &=&\sum_{n= \parallel V_{i-1} \parallel+1}^N \sum_{F \in \mathcal{A}(n)} I^{(i)}(n) \sum_{{\bm \pi} \in\Pi(F/\emptyset)} (n-1)!\rho^{n-1}e^{-\delta n(n-1)/2}\beta_{\pi_1} \sum_{i=0}^n\alpha^{(n)}_i( \bm{a} )    \mathbb{E}\left[e^{- a_i {\mathcal I}(0,t_k)}\right] \\
 &=&\sum_{n=\parallel V_{i-1} \parallel+1}^N  I^{(i)}(n)  (n-1)!\rho^{n-1}e^{-\delta n(n-1)/2}\sum_{i=0}^n\alpha^{(n)}_i(\bm{a})    \mathbb{E}\left[e^{- a_i {\mathcal I}(0,t_k)}\right]
 \sum_{F \in \mathcal{A}(n)}\sum_{{\bm \pi} \in\Pi(F/\emptyset)}\beta_{\pi_1} \\
 &=&\sum_{n=\parallel V_{i-1} \parallel+1}^N  I^{(i)}(n)  (n-1)!\rho^{n-1}e^{-\delta n(n-1)/2}\sum_{i=0}^n\alpha^{(n)}_i(\bm{a})    \mathbb{E}\left[e^{- a_i {\mathcal I}(0,t_k)}\right] a_0 \frac{(N-1)!}{(N-n)!},\\
&=&  \sum_{i=0}^N \mathbb{E}\left[e^{- a_i {\mathcal I}(0,t_k)}\right] a_0\sum_{n=\max(i,\parallel V_{i-1} \parallel+1)}^N  I^{(i)}(n)  \frac{(n-1)!(N-1)!}{(N-n)!}\rho^{n-1}e^{-\delta n(n-1)/2}\alpha^{(n)}_i(\bm{a}), \\
 &=&  \sum_{i=0}^{N-1} \mathbb{E}\left[e^{- a_i {\mathcal I}(0,t_k)}\right] \Gamma_i + a_0 \left( (N-1)!\right)^2 \rho^{N-1}I^{(i)}(N)e^{-\delta N(N-1)/2}\alpha^{(N)}_N(\bm{a}),
\end{eqnarray*}
where $(\bm{a}) := (a_0, a_1,\cdots,a_n)$ and $\mathcal{I}(0,t) = \int_0^t \Phi(u, Y_u) \dd u$. Since $a_i >0$ for all $i=0,1,\cdots,N-1$, by letting $t_k\rightarrow +\infty$, we derive
$$
 1= \lim_{t_k\rightarrow +\infty} \mathbb{E}[L^{(i)}(X_{t_{k}})] = a_0 \left( (N-1)!\right)^2 \rho^{N-1}I^{(i)}(N)e^{-\delta N(N-1)/2}\alpha^{(N)}_N(\bm{a}).
$$
The desired result is then obtained.
\end{proof}

\subsection{Proof of Proposition \ref{propexamplenonhomo}}
\label{appen_propexamplenonhomo}

\begin{proof}
Recall the definition of $\widehat{\Lc}^{\bm \pi}(n)$ in \eqref{eqn_Lc_hat}. Due to the contagion structure of NCM model, each obligor will only impact its two nearest neighbors. Hence,     $\widehat{\Lc}^{\bm \pi}(n)$ is 	
non-zero only if $F$ is a consecutive sequence of the circle  $\{1\rightarrow 2\rightarrow 3 \rightarrow \cdots \rightarrow N\rightarrow 1\}$.
Denote by $S(i)$ the sequence which has $n$ elements and starts with $1+i\% N$, i.e., $S(i) = \{1+i\% N, 1+(i+1)\% N, \cdots, 1+(i+n-1)\% N\}$. Here, $\%$ stands for the residue of two integers. We have:
\begin{align}
  \sum_{\pi\in\Pi( {  S(i)/\emptyset })}\widehat{\Lc}^{\bm  \pi}(n) &= \; \sum_{{\bm  \pi} \in\Pi(S(i)/\emptyset)} \sum_{j=0}^{n-1} \mathds{1}_{\{\pi_1 = 1+(i+j)\% N\}}\widehat{\Lc}^{\bm  \pi}(n) \\
  &= \; e^{-\delta n(n-1)/2}   \sum_{j=0}^{n-1} \beta_{1+(i+j)\%N} \, C_{n-1}^j \, \mathfrak{p}^{n-1-j}\mathfrak{q}^j,
\end{align}
where $C_n^k$ is the combination number of taking $k$ distinct elements out of $n$ elements.

Using the above result, we derive
\begin{eqnarray*}
  \sum_{F \in \mathcal{A}(n)} \sum_{{\bm \pi} \in\Pi(F/\emptyset)}\widehat{\Lc}^{\bm  \pi}(n) &=& \sum_{i=1}^N \sum_{{\bm \pi} \in\Pi(S(i)/\emptyset)}\widehat{\Lc}^{\bm  \pi}(n) = e^{-\delta n(n-1)/2}\sum_{i=1}^N \sum_{j=0}^{n-1} \beta_{1+(i+j)\%N} C_{n-1}^j \mathfrak{p}^{n-1-j} \mathfrak{q}^j \\
  &=& e^{-\delta n(n-1)/2}   \sum_{i=1}^N \beta_i    \sum_{j=0}^{n-1} C_{n-1}^j \mathfrak{p}^{n-1-j}\mathfrak{q}^j = {  \bar{a}_0} e^{-\delta n(n-1)/2}   (\mathfrak{p}+\mathfrak{q})^{n-1}.
\end{eqnarray*}
Notice that we have
 \begin{eqnarray*}
   \overline{\Lc}_{F^{\bm{\pi}}_0} &=& \overline{\Lc}_{\emptyset} = \sum_{i=1}^{N} \beta_i = \overline{a}_0, \ \  \overline{\Lc}_{F^{\bm{\pi}}_N} = \overline{\Lc}_{\cal N} = 0 = \overline{a}_N,\\
    \overline{\Lc}_{F^{\bm{\pi}}_k} &=& \sum_{i\in (F^{\bm{\pi}}_k)^c} \widehat{\Lc}_{F^{\bm{\pi}}_k}(i)    = e^{-\delta k} \sum_{i\in (F^{\bm{\pi}}_k)^c}\sum_{j\in F^{\bm{\pi}}_k}\rho_{ji}= e^{-\delta k}(\mathfrak{p}+\mathfrak{q}) := \overline{a}_k, \ \mbox{for all } k=1,2,\cdots, N-1.
 \end{eqnarray*}

Now we are ready to calculate $\mathbb{E}[L^{(i)}(X_{t_{k}})]$  from assertion (ii) of Proposition \ref{maintheoremmarkovG}
\begin{align}
\mathbb{E}[L^{(i)}(X_{t_{k}})]  &= \; \sum_{n=\parallel V_{i-1} \parallel +1}^N \; \sum_{F \in \mathcal{A}(n)} I^{(i)}(n) \sum_{{\bm \pi} \in\Pi(F/\emptyset)} \widehat{\Lc}^{\bm \pi}(n)\sum_{i=0}^n\alpha^{(n)}_i(\bm{\bar{a}}_n)    \mathbb{E}\left[e^{- \overline{a}_i {\mathcal I}(0,t_k)}\right] \\
 &= \;\sum_{n= \parallel V_{i-1} \parallel +1}^N  I^{(i)}(n)\sum_{i=0}^n\alpha^{(n)}_i(\bm{\bar{a}}_n)    \mathbb{E}\left[e^{- \overline{a}_i {\mathcal I}(0,t_k)}\right] \sum_{F \in \mathcal{A}(n)} \sum_{{\bm \pi}\in\Pi(F/\emptyset)} \widehat{\Lc}^{\bm \pi}(n) \\
 &= \; \sum_{n= \parallel V_{i-1} \parallel +1}^{N}  I^{(i)}(n) e^{-\delta n(n-1)/2} \, \overline{a}_0 (\mathfrak{p}+\mathfrak{q})^{n-1} \sum_{i=0}^n\alpha^{(n)}_i(\bm{\bar{a}}_n)    \mathbb{E}\left[e^{- \overline{a}_i {\mathcal I}(0,t_k)}\right] \\
 &= \; \overline{a}_0 \sum_{n= \parallel V_{i-1} \parallel +1}^{N-1}  I^{(i)}(n) e^{-\delta n(n-1)/2}  (\mathfrak{p}+\mathfrak{q})^{n-1} \sum_{i=0}^n\alpha^{(n)}_i(\bm{\bar{a}}_n)    \mathbb{E}\left[e^{- \overline{a}_i {\mathcal I}(0,t_k)}\right]  \\
 & \quad + \overline{a}_0 I^{(i)}(N)   e^{-\delta N(N-1)/2}(\mathfrak{p}+\mathfrak{q})^{N-1}\left( \sum_{i=0}^{N-1} \alpha^{(N)}_i(\bm{\bar{a}}_N)     \mathbb{E}\left[e^{- \overline{a}_i {\mathcal I}(0,t_k)}\right]  + \alpha^{(N)}_N(\bm{\bar{a}}_N) \right),
\end{align}
where in the last equality above, we used the fact $\overline{a}_N = 0$, and $(\bm{\bar{a}}_n) :=(\overline{a}_0,\overline{a}_1,\cdots \overline{a}_n) $ for a positive integer $n$. Since $\overline{a}_i>0$ for all $i=0,1,\cdots,N-1$, by letting $t_k\rightarrow +\infty$, it gives
$  1 = \lim_{t_k\rightarrow +\infty}  \mathbb{E}[L^{(i)}(X_{t_{k}})] =  \overline{a}_0 I^{(i)}(N)   e^{-\delta N(N-1)/2}(\mathfrak{p}+\mathfrak{q})^{N-1} \alpha^{(N)}_N(\bm{\bar{a}}_N).$  This completes the proof.
\end{proof}

\bibliographystyle{apalike}
\bibliography{CDO-reference}

\end{document}